\documentclass[twoside,11pt]{article}

\usepackage{jair, theapa, rawfonts}
\usepackage{amssymb}
\usepackage{enumerate}
\usepackage{graphicx}
\usepackage{comment}
\usepackage{amsmath,amsthm}
\usepackage{thmtools}
\usepackage{thm-restate}
\newtheorem{Theorem}{Theorem}

\newcommand{\leaveout}[1]{}
\ShortHeadings{Multidefender Security Games}
{Lou, Smith, \& Vorobeychik}

\allowdisplaybreaks

\begin{document}

\title{Multidefender Security Games}

\author{\name Jian Lou \email jian.lou@vanderbilt.edu\\
       \addr Electrical Engineering and Computer Science, Vanderbilt University
       \AND
       \name Andrew M. Smith \email amsmi@ucdavis.edu\\
       \addr Computer Science, University of California at Davis and
       \addr Sandia National Laboratories
       \AND
       \name Yevgeniy Vorbeychik \email yevgeniy.vorobeychik@vanderbilt.edu\\
       \addr Electrical Engineering and Computer Science, Vanderbilt University}
\maketitle
\begin{abstract}
Stackelberg security game models and associated computational tools
have seen deployment in a number of high-consequence security
settings, such as LAX canine patrols and Federal Air Marshal
Service. These models primarily focus on isolated systems where only
one defender is present, despite being part of a more complex system
with multiple players. Furthermore, many real systems such as
transportation networks and the power grid exhibit interdependencies
between targets and, consequently, between decision makers jointly
charged with protecting them. 
In order to understand such multi-defender strategic interactions
present in security, we investigate game theoretic models of
security games with multiple defenders.
Unlike most prior analysis, we specifically focus on the situations in
which each defender must protect multiple targets, so that
even a single defender's best response decision is, in general, highly
non-trivial.
We start with an analytical investigation of multi-defender security games with
independent targets, offering an equilibrium and price-of-anarchy analysis
of three models with increasing generality.
In all models, we find that the defenders have the incentive to over-protect the targets, at times significantly. Additionally, in the simpler models, we find that the price of anarchy is unbounded, linearly increasing both in the number of defenders and the number of targets per defender.
Considering interdependencies among targets, we develop a novel
mixed-integer linear programming formulation to compute a defender's
best response, and make use of this formulation in approximating Nash
equilibria of the game.
We apply this approach towards computational strategic analysis of
several 
models of networks representing interdependencies, including
real-world power networks.
Our analysis shows how network structure and the probability of failure
spread determine the propensity of defenders to over- or under-invest
in security.
\end{abstract}

\section{Introduction}
\label{sec:intro}
Security, physical and cyber, has come to the forefront of national
attention, particularly after 9/11.
Among the variety of approaches that are used to tackle security
problems, from risk analysis to red teaming, game theory has had a
significant impact, with tools based on game theoretic analysis having
been deployed in LAX
airport to schedule canine
patrols~\shortcite{Paruchuri08,Jain08:Bayesian,Pita09:Using}, by Federal
Air Marshall Service (FAMS) to schedule the air
marshals~\shortcite{Kiekintveld09:Computing,Jain10,Jain10a}, and by the US
Coast Guard to schedule boat patrols~\shortcite{Shieh12}.
All of these deployments, and numerous other related efforts, have
cast security as a Stackelberg game between a single defender and an attacker, in
which the defender leads (i.e., acts first), choosing a probability
distribution over defense actions, and the attacker, upon learning
this probability distribution, chooses a response~\shortcite{Conitzer06:Computing}.
In many cases, the attacker is modeled as a rational agent who selects
an optimal response and, in the many applications that compute a Strong
Stackelberg equilibrium, an attacker is often assumed to break ties in the
defender's favor~\shortcite{Paruchuri08,Korzhyk11}.

A crucial assumption that all these efforts have in common is that a single defender is responsible for all the targets that need protection, and that she has control over all of the security resources. However, there are many domains in which there are multiple defender agencies who are in charge of different subsets of all targets. In practice, numerous parties are responsible for security; indeed, the fact that the basic framework has been deployed by different entities and agencies makes this manifest already.
If security decisions made by different parties were entirely independent, both from the defender's and the attacker's perspective, a single-defender model would be entirely satisfactory. However, the assets protected by different entities are typically
interdependent, or, more generally, have value to others who are not involved in security decisions. Additionally, attackers, insofar as they may target different sectors under the charge of different defenders, are resource constrained, implicitly coupling otherwise independent targets.

An important motivating application for our multidefender security game is security and reliability in the power grid. Independent System Operators (ISOs) and profit-driven independent utility operators are largely responsible for operating and controlling subsystems of the entire grid ~\shortcite{RAP11a}. These operators are held responsible for the reliability of their system, and thus have independent, and possibly even competing, goals with neighboring ISOs. As such, their security decisions are made independently, despite the interdependencies present between subsystems. 
As a result of this organization, cascading failures in the power grid can present a great threat to the entire system, even when an explicit attacker is not present. This problem is exacerbated by the fact that components in the grid are controlled by multiple entities and are also dependent on other independently operated utility networks (water, communications, natural gas, etc.).

We extend the previous Stackelberg game models in two ways:
\begin{enumerate}
	\item an analytic equilibrium and price of anarchy (PoA) characterization of $3$ multidefender security scenarios, in which we assume \emph{homogeneous and independent valuations} of the targets for each defender; and
	\item a computational analysis leveraging a novel mixed-integer linear program (MIP) approach for computing a defender's best response, combined with a novel heuristic method for approximating equilibria in interdependent multi-defender security games with heterogeneous targets. 
\end{enumerate}
In case $1$ where there are multiple defenders, and the values of the targets are \emph{independent} and \emph{homogeneous} among the defenders, our analysis is focused on three models of such multi-defender games (each varying in their level of generality). We show that a Nash equilibrium among defenders in this two-stage game model need not always exist, even when the defenders utilize randomized strategies (i.e., probability distributions over target protection levels); this is distinct from a model in which the attacker moves simultaneously with the defenders, where a mixed strategy equilibrium is guaranteed to exist.
When an equilibrium does exist, we show that the defenders protect all of their targets with probability 1 in all three models, whereas the socially optimal protection levels are generally significantly lower.
When no equilibrium exists, we characterize the best approximate Nash equilibrium (that is, one in which defenders have the least gain from deviation), showing that over-investment is substantial in this case as well.
Our \emph{price of anarchy (PoA)} analysis, which relies on the unique equilibrium when it exists, and the approximate equilibrium otherwise, demonstrates a surprising finding: whereas PoA is unbounded in the simpler models, increasing linearly with the number of defenders, the more general model shows this to be an atypical special case achieved when several parameters are exactly zero.
More generally, PoA is bounded by a constant. 

For case $2$, we introduce interdependencies between targets. 
Because closed-form analysis in this setting is intractable, we propose a novel mixed-integer linear programming approach combined with a novel heuristic method to approximate equilibrium behavior. Unlike other multi-defender models (e.g.,~\shortcite{Kunreuther03,Chan12,Bachrach12,Acemoglu13}), our approach maintains the typical complexity of \emph{individual defender} decision process in the multi-defender framework, with each defender responsible for securing \emph{many}, possibly interdependent, targets. Our setup gives rise to two competing externalities of security decisions: a \emph{positive} externality, where greater
security implies reduced contagion risk to other defenders, and a \emph{negative} externality, which arises because high security by one player pushes the attacker to attack someone else's assets. We study the impact of competing externality effects of defense on the resulting Nash equilibrium
outcomes as a function of network topology (using both synthetic and
real networks), interdependent risk, and the level of system
decentralization. One of our key findings is that the impact of system
decentralization on security and welfare can be non-monotonic when
systems are highly interdependent: high
levels of decentralization can yield near-optimal outcomes, even as
moderate decentralization results in significant
underinvestment. With weak interdependencies, on the other hand, an
increasingly decentralized system tends more strongly to over-invest in security.

The remainder of our paper is outlined as follows. In Section $2$, we give an overview of related work. In Section $3$, we briefly outline the definitions and solution concepts of our independent and interdependent multi-defender security games, respectively. Section $4$ provides an equilibrium and PoA analysis of the homogeneous, independent security game models. Section $5$ further explains the interdependent multi-defender model, and presents results on well-studied synthetic networks, as well as on real-world power grid networks. 

\section{Related Work}
\label{sec:previouswork}

Our work, like much work in the recent security game literature,
builds on the notion of Stackelberg games~\cite{Osborne94}, which model commitment in strategic settings.
The first thorough computational treatment of randomized (mixed
strategy) commitment was due to~\citeA{Conitzer06:Computing}.
In this line of work, of greatest relevance to our effort are
multiple-leader Stackelberg
games~\shortcite{Sherali84,DeMiguel09,Leyffer07,Rodoplu10,Kulkarni14,Sinha14}.
In many cases, these approaches leverage specialized problem
structure, and are not immediately applicable to our setting.
In particular,~\citeA{Sherali84} and~\citeA{DeMiguel09} focus on relatively simple models with firms
setting production quantity (a single variable), aiming to maximize
profit.
Both show existence and uniqueness of equilibria in their setting, and
leverage these characterization results to obtain solutions to the
games.
Similarly,~\citeA{Rodoplu10} consider a relatively
simple model of network competition in which leaders are nodes setting
prices for packets transmitted through them; again, each leader only
sets a single variable, the utility functions are problem-specific,
and algorithms are specialized to the particular problem structure (and
are inapplicable to our setting). 

\citeA{Sinha14} propose an evolutionary algorithm for
solving bi-level Stackelberg problems, but their problem structure is
also highly specific to the domain of interest (firms choosing
production, investment, and marketing, and maximizing profit), and the
evolutionary algorithm leverages significant simplifications, such as
the assumption that the market eventually clears.
\citeA{Leyffer07} present a very generic
multi-leader multi-follower setting and solution framework in the
context of shared complementarity constraints (which is the case for
our problem, where a single follower attacks a single target), but rely on separability of objective
functions in leader and follower variables, the assumption that does
not obtain in our setting (in addition, their approach only scales to
2-4 leaders, whereas we are able to approximately solve games with 64 leaders).
\citeA{Kulkarni14} offer a deep theoretical
treatment of a relatively broad class of multi-leader multi-follower
games, but much of their analysis and positive results are restricted
to potential games, and they do not offer specific algorithmic suggestions.
Like us, they leverage shared constraints to resolve the issue of
incompatible leader assumptions about the follower tie-breaking behavior.

Our point of departure is a class of Stackelberg games specifically
pertinent to security: commonly, these are simply known as
\emph{security games}~\shortcite{Korzhyk11,Paruchuri08,Jain10a,Vorobeychik11}.
In these games, a single defender allocates a set of resources among
potential targets of attack in a randomized fashion (that is, the
defender commits to a probability distribution over resource-to-target
mappings), with an attacker choosing a single target
to attack after observing the defender's strategy.
Almost universally in this domain, a Strong Stackelberg equilibrium
(SSE) is a solution concept of choice.
In SSE, the follower (attacker) is assumed to break ties in the
defender's favor.
As we will see below, this solution concept presents conceptual and technical problems
in a multi-defender setting.

A similar, though mostly orthogonal, line of work are network
interdiction problems~\shortcite{Cormican98,Woodruff03}, in which a leader
attempts to interdict a network over which the follower subsequently
solves a variation of a network flow problem.
Unlike the literature on security games, as well as our setting, network interdiction problems
are almost universally zero-sum (minimax).

Another somewhat related line of work considers the problem of
coordination and teamwork among multiple defenders in a purely
cooperative setting~\cite{Jiang13,Shieh14}.
This work, however, is entirely unlike ours: in particular, our primary
focus is on the impact of \emph{incentive misalignment} among the
defenders with different (though certainly related) motivations,
rather than coordination issues and teamwork.
While often effective coordination among multiple defenders can be
achieved, just as often (if not predominantly) decentralization of
decision making processes and resources inherently give rise to
distinct, and often conflicting, incentives among defenders.

Among the earliest multi-defender models is the literature on \emph{interdependent
  security games}~\shortcite{Kunreuther03}, in which interactions among multiple defenders are
modeled as an $n$-player, $2$-action game, where a player decides
whether to invest in security; however, no attacker is considered. More recently, time-dependent scenarios where coordination of defender resources amongst multiple defenders is assumed have been studied using Markov decision processes~\shortcite{Shieh14a}. Since total cooperation is assumed, this model effectively reduces to a single defender game in which the defender controls all resources.
A recent extension, \emph{interdependent defense
  games}~\shortcite{Chan12}, does consider an attacker who acts \emph{simultaneously} with
the defenders, rather than after observing the joint defense
configuration, as in our model. Interdependent defense games have also been studied in the context of traffic infrastructure defense~\shortcite{Alderson11a}.
Two recent efforts studying multi-defender games explicitly model interdependence among
targets through a probabilistic contagion process~\shortcite{Bachrach12,Acemoglu13}.
Like our paper, they consider attackers who observe the
joint defense prior to making a decision, but each defender is
restricted to secure a single node, and strategy space is assumed to
be continuous.
\citeA{Vorobeychik11} is, to our knowledge, the only other attempt to
study strategic settings related to security in which each player's
decision space is combinatorial.
However, this work does not consider a strategic attacker.


\section{Multidefender Models}
\label{sec:models}
Our modeling effort proceeds in four steps, each generalizing the previous.
As we see below, each generalization step reveals new and surprising insights about the multi-defender security setting, allowing us to appreciate the fundamental incentive forces. The first three models deal with homogeneous, independent targets and will be analyzed exactly using Nash equilibrium and price of anarchy (PoA) analysis. Our final model introduces interdependencies, and will be analyzed using computational methods.

\subsection{The Baseline Model}

We start with a model which most reflects the related literature: in particular, this model involves $n$ defenders and a single attacker, with each defender engaged in protecting a single target.
Each target has the same value to the defender $v>0$. 
We suppose that the defender has two discrete choices: to protect the target, or not.
In addition, the defender can randomly choose among these; our focus is on these \emph{coverage probabilities} (i.e., the probability of protecting, or covering, the target), which we denote by $q_i$ for a given defender $i$.
The attacker is strategic, can observe the defenders' coverage
probabilities, and chooses a target that maximizes the damage. We assume that attacker is indifferent among the targets, and attacks the target with the lowest coverage probability, breaking ties uniformly at random. 
In a given scenario, for all defenders, the attacker's strategy is a vector of probabilities $P=<p_1, p_2, ..., p_n>$, where $p_i$ is the probability of attacking the target protected by defender $i$, with $\sum_{i=1}^{n} p_i=1$. 

We assume that if the attacker chooses to attack a target
corresponding to defender $i$ and defender $i$ chooses to protect the
target, then the utility of both is $0$, and if the attacker attacks
the target but it is not protected, then the utility of the defender
is $-v$ while the attacker's utility is $v$. If a defender chooses to cover a target, it will incur a cost $c>0$. 
Additionally, we assume that targets are independent, i.e., if
defender $i$ is successfully attacked, all other defenders $j \ne i$
receive 0 utility.
We can thus define the expected utility of a defender $i$ as $$u_i=p_iu_i^{a}+(1-p_i)u_i^{u},$$ where
$u_i^{a}$ is the utility of $i$ if it is attacked, and $u_i^{u}$ is the utility of $i$ if it is not attacked.
By the assumptions above,
  $$ u_i^{a}=-(1-q_i)v-q_ic =-v+q_i(v-c) $$ 
 $$u_i^{u}=-q_ic.$$

\subsection{The Multi-Target Model}

Our key conceptual departure from related work is in allowing each defender to protect multiple targets, aligning it better with practical security domains.
Specifically, suppose that there are $n$ defenders, each protecting $k \ge 1$ targets. 
Then the strategy of defender $i$ will be a vector $<q_{i1}, q_{i2}, ... q_{ik}>$. 
The strategy profile of the attacker can then be described as a matrix of probabilities,
\[
\begin{pmatrix}
p_{11}& p_{12}&\ldots &p_{1k}\\
p_{21}& p_{22}& \ldots& p_{2k}\\
 \vdots&\vdots &\ddots&\vdots\\
 p_{n1}& p_{n2}& \ldots & p_{nk}\\
\end{pmatrix}
\]
in which $\sum_{i=1}^{n}\sum_{j=1}^{k} p_{ij}=1$ and $p_{ij}\geq 0$ for each $i$ and $j$. 
The expected utility of a defender $i$ in this model is  $$u_i=\sum_{j=1}^{k} p_{ij}u_{ij}^{a}+(1-p_{ij})u_{ij}^{u},$$ 
where $u_{ij}^{a}$ is the utility of target $j$ to defender $i$ if it is attacked, and $u_{ij}^{u}$ is the utility of target $j$ to $i$ if it is not attacked.
Using the notation introduced earlier, we have
  $$ u_{ij}^{a}=-(1-q_{ij})v-q_{ij}c =-v+q_{ij}(v-c) $$ 
 $$u_{ij}^{u}=-q_{ij}c.$$

\subsection{The General Model}

Generalizing further, we assume that if the attacker chooses to attack a target protected by defender $i$ and the defender chooses to protect the target, then the utility of the target to defender $i$ is $U^c$, and if the attacker attacks the target but it is not protected, then the utility of the target to the defender is
$U^u$. It is reasonable to assume that $U^c\geq U^u$. If the target of defender $i$ is not attacked, then we assume that the utility of the target for defender $i$ is $\Omega\geq U^c$. Other assumptions are the same as those in the multi-target model.  
In the general model, therefore,
   $$ u_{ij}^{a}=q_{ij}U^c+(1-q_{ij})U^u-q_{ij}c=(U^c-U^u-c)q_{ij}+U^u $$ 
 $$u_{ij}^{u}=\Omega-q_{ij}c.$$

\subsection{The Interdependent Model}
The previous three models featured three important restrictions: first,
that target values are homogeneous, second, that targets are
independent, and third, that defenders protect the same number of targets.
We now relax these restrictions.
Suppose that a defender $i$ can choose from a finite set $O$ of
security configurations for each target $j \in T_i$, where $T_i$ is
the set of targets under $i$'s direct influence.
Let $T$ be the set of \emph{all} targets, that is, $T = \cup_i T_i$,
with $|T| = m$.
A configuration $o \in O$ for target $j \in T_i$ incurs a cost $c_{j}^o$
to the defender $i$.  
If the attacker attacks a target $j \in T$ while
configuration $o$ is in place, the expected value to a defender $i$ is
denoted by $U_{i,j}^o$, while the attacker's value is $V_{j}^o$.
We assume in this model that each player's utility depends only on the
target attacked and its security
configuration~\shortcite{Kiekintveld09:Computing,Letchford12}.
We denote by $q_{i,j}^o$ the probability that the defender $i$ chooses
$o$ at target $j \in T_i$.

While the problem we study assumes that that the utility of any player
for a given target depends only on its security configuration $o$,
there is a rather natural way to model interdependencies while
retaining this structure, proposed by~\citeA{Letchford12}.
Specifically, suppose that dependencies between targets are represented by a graph $(T,E)$,
with $T$ the set of targets (nodes) as above, and $E$ the set of edges
($j,j'$), where an edge from $j$ to $j'$ means that 
a successful attack on $j$ may have impact on $j'$.  Each target $j$ has
associated with it a value, $v_{ij}$, for the defender $i$, which is
the loss to $i$ if $j$ is affected (e.g., compromised, broken).  
The security configuration determines the probability $z_{j}^o(j)$ that
target $j$ is affected if the attacker attacks it
\emph{directly} and the defense configuration is $o$.  We model the
interdependencies between the nodes as independent cascade contagion~\shortcite{Kempe03,Letchford12}.
The contagion proceeds starting at an attacked node $j$, affecting its network
neighbors $j'$ each with probability $r_{j,j'}$, the contagion then
spreads from the newly affected nodes $j'$ to their neighbors, and so on.  The contagion can
only occur one time along any network edge, and once a node is affected it stays affected
through the diffusion process.  Each player's valuation for each target is then updated based on the probability of a failure cascading to one of the player's owned targets.

\subsection{The Weakness of Strong Stackelberg Equilibrium}
By far the most important solution concept in Stackelberg security
games is a \emph{Strong Stackelberg equilibrium (SSE)}.
A SSE is characterized by an assumption that the attacker breaks ties
in defender's favor.
When there is a single defender, this is well defined, and
quite reasonable when the defender can commit to a mixed strategy: a
slight adjustment in the defense policy will force the attacker to
strictly prefer the desired option, with little loss to the defender.
As we now illustrate, however, SSE is fundamentally problematic in a
multi-defender context, because the notion of ``breaking ties in
defender's favor'' is no longer well defined in general, as we must
specify \emph{which} defender will receive the favor.

To see concretely what goes wrong, consider the example in
Figure~\ref{F:SSE_example}. In this example there are two defenders,
\begin{figure}
\centering
\includegraphics[width=1.5in]{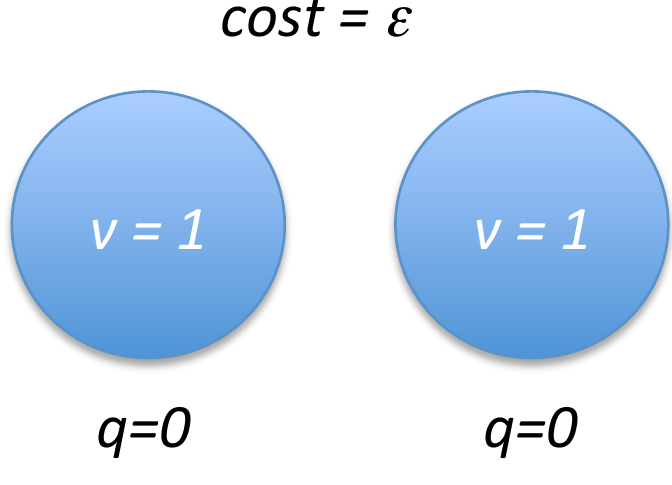}
\caption{Example of a problem with a SSE in a multi-defender setting.}
\label{F:SSE_example}
\end{figure}
one who defends the target
on the left, while the other defends the target on the right.
Both defenders value their respective targets at $1$, and have no
value for the counterpart's target.
The cost of defending each target is $0 < \delta \ll 1$.
Now, consider a strategy profile in which $q_j = 0$ for both targets
$t$, and let us focus on the best response of the first (left)
defender.
If this defender attempts to compute an SSE by fixing the strategy of
the second player, he perceives his utility under the current strategy
profile to be $1$, since he would assume that the attacker breaks ties in his
favor and, thus, attacks the defender on the right.
By the same logic, the defender on the right will assume that the
attacker will attack his counterpart, and perceive $q_j = 0$ to be the
best response.
Since the attacker actually attacks one of them, the best response of
the defender being attacked is to defend with a small probability, pushing
the attacker towards the other target.
What goes wrong here is that both players assume that the
attacker attacks the other (breaks ties in their favor), which is
inconsistent with the assumption that the attacker will certainly
attack \emph{some} target.\footnote{The problem we observe is similar to the issue of inconsistent
conjectures the leaders in a multi-leader Stackelberg game could have
about follower behavior noted by Kulkarni and
Shanbhag~\cite{Kulkarni14}. The solution we propose below---ASE---has the
effect of imposing a shared constraint, the idea introduced by Kulkarni and
Shanbhag generically. We note, however, that our concern here is not
merely the fact that inconsistent conjectures lead to disequilibrium;
rather, they could lead to nonsensical equilibria!}


\subsection{Solution Concepts}

Since the classic (two-player) SSE solution concept used in
Stackelberg security games does not conceptually extend to be an individual
defender best-response problem in the multi-defender setting, we need to consider an alternative.
One option is to compute an arbitrary subgame perfect equilibrium.
However, we wish to impose a natural constraint on the solution
concept that the attacker's best response be computed consistently for
any joint defense policy, just as it is in a SSE (in other words, we wish to fix a
tie-breaking rule).
One natural tie-breaking rule is that the attacker chooses a target
uniformly at random from the set of all best responses.
We call the corresponding solution concept (which is a refinement of
the subgame perfect equilibrium of our game) the \emph{Average-case
  Stackelberg Equilibrium (ASE)}.
The crucial property of this solution concept that we desire is that
the attacker's behavior presumed by a defender's best response
problem is independent of that defender's identity, a property that
SSE violates. 
As we demonstrate below, ASE is not guaranteed to exist, in which case we focus on $\epsilon$-equilibria, in which no defender gains more than $\epsilon$ by deviating; in particular, we will consider $\epsilon$-equilibria with the smallest attainable $\epsilon$. 

To measure how the efficiency of the game degrades due to selfish behavior of the defenders, we consider \emph{Utilitarian Social Welfare} and \emph{$(\epsilon)$-Price of Anarchy} in our paper. \emph{Utilitarian Social Welfare} is the sum of all defenders' payoffs. For the smallest attainable $\epsilon$, we define $\epsilon$-Price of Anarchy ($\epsilon$-PoA) as follows:
\begin{equation}
\epsilon\text{-}PoA= \frac{SW_O}{SW_E} \nonumber
\end{equation}
where $SW_O$ is the optimal (utilitarian) social welfare that can be
obtained (i.e., if there was a single defender), and $SW_E$ is the
worst-case (utilitarian) social welfare in $\epsilon$-equilibrium. An
underlying assumption of this definition is that the value of $SW_O$
and $SW_E$ are both positive. If they are both negative, then
$\epsilon$-PoA will be the reciprocal of above equation. Note that the
ordinary \emph{Price of Anarchy} is a special case of $\epsilon$-Price
of Anarchy with $\epsilon=0$. 

\section{Equilibrium Analysis of Independent Multi-Defender Security Games}
\label{sec:indep-analysis}

In this section, we consider scenarios in which the values of the targets are \emph{independent} and \emph{homogeneous} among the defenders. Our equilibrium and Price of Anarchy analysis will show that 
a Nash equilibrium among defenders in the two-stage game model
(equivalently, ASE) need not always exist, even when the defenders
utilize randomized strategies (i.e., probability distributions over
target protection levels). For cases when there is no Nash
equilibrium, we make use of approximate Nash (ASE) equilibrium and the
associated ($\epsilon$)-Price of Anarchy.





\subsection{The Baseline Model}

Our first result presents necessary and sufficient conditions for the
existence of a Nash equilibrium among defenders in the baseline model, and characterizes it when it does exist.

\begin{Theorem}
In the \emph{Baseline model}, Nash equilibrium exists \emph{if and only if} $v\geq c$. In this equilibrium all targets are protected
with probability 1.
\end{Theorem}

\begin{proof}
Firstly, we claim that Nash equilibrium among defenders can appear \emph{only if} all targets have the same coverage probability $q$.
Otherwise, some defender $j$ who has probability 0 of being attacked has the incentive to decrease her $q_j$. To find the Nash equilibria, we need only consider strategy profiles in which all targets have the same coverage probability. 

When all defenders use the same coverage probability $q$, each
defender's expected utility is $$u=\frac{(v-cn)q-v}{n}.$$
If $q<1$, some defender $i$ could increase $q$ to $q+\delta$, where
$\delta$ is a small positive real number, to avoid being attacked, and
receive utility $u'=-(q+\delta)c$, so that $$u'-u=\frac{v(1-q)-nc\delta}{n}.$$
As $\delta$ can be arbitrarily small, $u'-u>0$ when $q<1$.
Consequently, when $q<1$, a defender always has an incentive to
deviate, which implies that the only possible Nash equilibrium can be
for all players to play $q=1$.

When all defenders use coverage probability $q=1$, each earns an expected utility of 
$$u=-c.$$ 
If a defender $i$ decreases her coverage probability to $q'<1$, then
she will be attacked with probability 1, and receive expected utility
$u'=-v+q'(v-c) $, so that $$u'-u=(v-c)(q'-1).$$
If $v\geq c$, then $u'-u\leq 0$, and $q=1$ is indeed a Nash equilibrium.
If $v<c$, however, $u'-u> 0$, which implies that a Nash equilibrium
does not exist.
\end{proof}

Thus, if a Nash equilibrium does exist, it is unique, with all defenders always protecting their targets.
But what if the equilibrium does not exist?
Next, we characterize the (unique) $\epsilon$-equilibrium with the minimal $\epsilon$ that arises in such a case.
We will use this approximate equilibrium strategy profile as a \emph{prediction} of the defenders' strategies.



\begin{Theorem}
In the \emph{Baseline model}, if $v<c$, the optimal $\epsilon$-equilibrium is for all defenders to cover their target with probability $\frac{v}{c}$. The corresponding $\epsilon$ is $\frac{v(c-v)}{cn}$.
\end{Theorem}

\begin{proof}
We firstly consider strategy profiles in which all targets have the
same probability $q$ of being protected. Then each defender's expected
utility is
$$u=\frac{(v-cn)q-v}{n}.$$

Suppose $0\leq q<1$. If a defender $i$ slightly increases $q$ to
$q+\delta_1$, she could receive a utility $u'=-(q+\delta_1)c$, with
$$u'-u=\frac{v(1-q)-nc\delta_1}{n}<\frac{v(1-q)}{n}.$$

Suppose $0<q\leq 1$. If a defender $i$ slightly decreases $q$ to
$q-\delta_2$, she could receive utility $u''=-v+(q-\delta_2)(v-c)$, with
$$u''-u=\frac{v(1-q)(1-n)+\delta_2 n(c-v)}{n}.$$
Since $\delta_2\leq q$, 
$$u''-u\leq \frac{v(1-q)(1-n)+qn(c-v)}{n}=\frac{v(1-q)}{n}+(qc-v).$$
Let $d_1=\frac{v(1-q)}{n}$,
$d_2=\frac{v(1-q)}{n}+(qc-v)$. If $q=0$, a defender could
deviate to increase utility by at most $\frac{v}{n}$. If $q=1$, a defender could
deviate to increase utility by at most $(c-v)$.
When $0<q\leq \frac{v}{c}$, we have that $d_2\leq d_1$, and
it is $d_1$-equilibrium. When $\frac{v}{c}<q<1$, we have that
$d_2>d_1$, and it is $d_2$-equilibrium. 

Putting everything together, there is an $\epsilon$-equilibrium with
\begin{equation}
\epsilon=
\begin{cases}
\frac{v(1-q)}{n}, &\text{if $0\leq q\leq \frac{v}{c}$;}\\
\frac{v(1-q)}{n}+(qc-v), &\text{if $\frac{v}{c}<q\leq 1$.}
\end{cases} \nonumber
\end{equation}
Moreover, $q=\frac{v}{c}$ gives us the unique minimal
$\epsilon=\frac{v(c-v)}{cn}$ when all defenders use the same coverage probabilities.

We now claim that the $\frac{v(c-v)}{cn}$-equilibrium \emph{could
  only} exist when all defenders play an identical coverage probability.
Suppose defenders use different coverage probabilities. Then there are
$\alpha$ defenders for $1\leq \alpha<n$ who have the same minimal probability $q'$ of protecting their targets. The expected utility for each defender among these $\alpha$ defenders is: 
$$u_e=\frac{(v-c\alpha)q'-v}{\alpha}.$$
When $\frac{v}{c}<q'\leq 1$, some defender $i$ among these $\alpha$
defenders could decrease her probability to $0$ to get a utility of $u_1=-v$, with
$$u_1-u_e=\frac{v(1-q')}{\alpha}+(q'c-v)>\frac{v(1-q')}{n}+(q'c-v) > \frac{v(c-v)}{cn}.$$
 When $0\leq q'\leq \frac{v}{c}$, some defender $i$ among these $\alpha$ defenders could increase her probability to $q'+\delta_3$ to get utility 
 $u_2=-(q'+\delta_3)c$ with
 $$u_2-u_e=\frac{v(1-q')-\alpha c\delta_3}{\alpha}>\frac{v(1-q')}{n}
 \ge \frac{v(c-v)}{cn},$$
where the first inequality follows because $\delta_3$ can be made
arbitrarily small.
\end{proof}

Armed with a complete characterization of predictions of strategic behavior among the defenders, we can now consider how this behavior related to socially optimal protection decisions.
Since the solutions are unique, there is no distinction between the notions of \emph{price of anarchy} and \emph{price of stability}; we term the ratio of socially optimal welfare to welfare in equilibrium as the price of anarchy for convenience.


First, we characterize the socially optimal outcome.
\begin{Theorem}
In the \emph{Baseline model}, the optimal social welfare $SW_O$ is 
\begin{equation}
SW_O=
\begin{cases}
-cn, &\text{if $v\geq cn$;}\\
-v, &\text{if $v<cn$.}
\end{cases} \nonumber
\end{equation}
\end{Theorem}
\begin{proof}[Proof sketch]
We firstly claim that we could get optimal social welfare \emph{only
  if} all defenders use the same coverage probability $q$. If their
coverage probabilities are different, and some defender $j$ has
probability $0$ of being attacked, we could decrease $q_j$ to improve
social welfare. Therefore we need only to consider identical coverage
probabilities $q$ in determining optimal social welfare. 
Welfare, as a function of symmetric coverage $q$ is
$$SW(q)=-v+q(v-c)+(n-1)(-c)=(v-cn)q-v.$$ 
When $v \ge cn$, $q=1$ is optimal, whereas $q=0$ is optimal otherwise,
giving the desired result.
\end{proof}
From this result, it is already clear that defenders systematically over-invest in security, except when values of the targets are quite high.
This stems from the fact that the attacker creates a \emph{negative externality} of protection: if a defender protects his target with higher probability than others, the attacker will have an incentive to attack another defender.
In such a case, we can expect a ``dynamic'' adjustment process with defenders increasing their security investment well beyond what is socially optimal.
To see just how much the defenders lose in the process, we now characterize the price of anarchy of our game.

If $v\geq c$, 
social welfare in the unique equilibrium with $q=1$ is $$SW_E=-cn.$$
The associated PoA is then
\begin{equation}
PoA=
\begin{cases}
1, &\text{if $v\geq cn$;}\\
\frac{nc}{v}, &\text{if $c<v<cn$.}
\end{cases} \nonumber
\end{equation}


Figure \ref{fig1} shows the relationship among PoA, the number of
defenders, and the ratio $\frac{c}{v}$. 
From the figure we can see that when number of defenders and $\frac{c}{v}$ are small (e.g. $n\leq$ 5 and $\frac{c}{v}=0.2$), the price of anarchy is close to $1$. Otherwise, PoA is unbounded, growing linearly with $n$.

\begin{figure}
\begin{center}
\includegraphics[width=120mm,height=67.5mm]{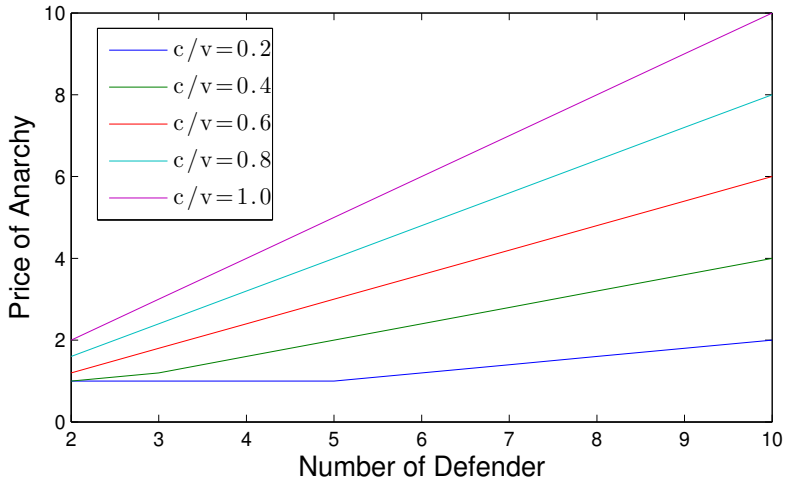}
\caption{Price of Anarchy when $v\geq c$}\label{fig1}
\end{center}
\end{figure}

When $v<c$, there is no Nash equilibrium. However, the optimal $\epsilon$-equilibrium features all defenders with the same coverage probability $\frac{v}{c}$ for their targets.  The corresponding Social Welfare is 
$$SW_E=(v-cn)\frac{v}{c}-v,$$ 
and the associated 
$\frac{v(c-v)}{cn}$-PoA
is
$\frac{cn+c-v}{c}$,
which is, again, linear in $n$.


\subsection{The Multi-Target Model}

Armed with observations from the model with a single target for each
defender, we now extend the model to a case not as yet considered in
the literature in a theoretical light: each defender protects a set of
$k$ targets.
This gives rise to a combinatorial set of possible decisions for each defender, so that even computing a best response is not necessarily easy.
Remarkably, we are able to characterize equilibria and approximate
equilibria in this setting as well. The proofs for this subsection are
in the appendix.

Our first result is almost a mirror-image of the corresponding result
in the baseline model: when a Nash equilibrium exists, all defenders
protect all of their targets with probability 1.
\begin{restatable}{Theorem}{EquMultiDefender}
\label{thm:Equ_Multi_Defender}
In the \emph{Multi-Target model}, Nash equilibrium exists \emph{if and
  only if} $v\geq kc$. In this equilibrium all targets are protected
with probability 1.
\end{restatable}


Next, we consider scenarios when $v<kc$, in which there is no Nash equilibrium. 
Our next result characterizes optimal (lowest-$\epsilon$) approximate Nash equilibria.
\begin{restatable}{Theorem}{AppMultiDefender}
\label{thm:AppMultiDefender}
In the \emph{Multi-Target model}, if $v<kc$, then in the optimal
$\epsilon$-equilibrium all targets are protected with probability $\frac{v}{kc}$. The corresponding $\epsilon$ is $\frac{v(kc-v)}{cnk}$.
\end{restatable}

Thus, as $n$ increases, the optimal approximate
equilibrium approaches a Nash equilibrium.
Figure \ref{fig3} illustrates the relationship between $\epsilon$ and
the number of targets each defender protects when $v=10$ and $c=1$.
In this figure, $\epsilon=0$ when $k\leq 10$, which means that
an exact Nash equilibrium exists; $\epsilon$ increases with $k$ when
$k>10$, but at a decreasing rate, converging to $\frac{v}{n}$ as $k
\rightarrow \infty$.

\begin{figure}
\begin{center}
\includegraphics[width=120mm,height=60mm]{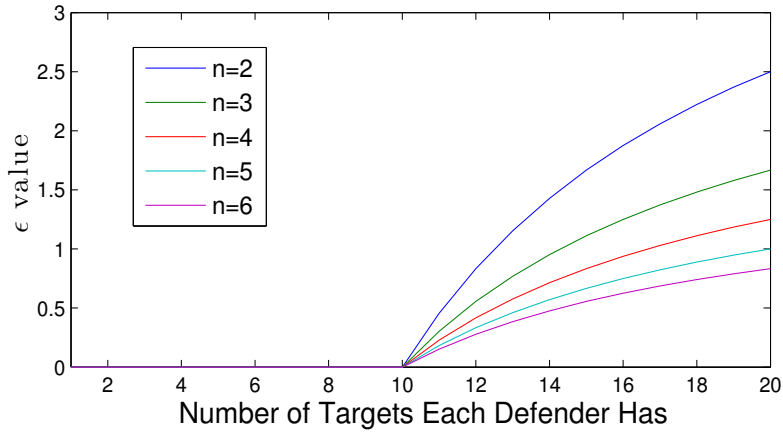}
\caption{$\epsilon$ value when $v=10,c=1$}\label{fig3}
\end{center}
\end{figure}

Finally, we characterize socially optimal welfare, and, subsequently, put everything together in describing the price of anarchy.
 
 \begin{restatable}{Theorem}{SocialMultiDefender}
 \label{thm:Social_Multi_Defender}
In the \emph{Multi-Target model}, the optimal social welfare $SW_O$ is
\begin{equation}
SW_O=
\begin{cases}
-cnk, &\text{if $v\geq cnk$;}\\
-v, &\text{if $v<cnk$.}
\end{cases} \nonumber
\end{equation}
\end{restatable}
 
  Thus, just as in the baseline case, the defenders will generally over-invest in security.

If $v\geq kc$, there is a unique Nash equilibrium with all targets
protected with probability 1. 
The corresponding social welfare is $$SW_E=-cnk$$
Because it is the only Nash equilibrium, the \emph{Price of Anarchy} is
\begin{equation}
PoA=
\begin{cases}
1, &\text{if $v\geq cnk$;}\\
\frac{nkc}{v}, &\text{if $ck\leq v<cnk$.}
\end{cases} \nonumber
\end{equation}

If $v<kc$, there is no Nash equilibrium. The optimal approximate
equilibrium features identical coverage probability of $\frac{v}{kc}$ for all targets.  The corresponding Social Welfare is 
$$SW_E=(v-cnk)\frac{v}{kc}-v,$$ 
and the associated $\frac{v(kc-v)}{cnk}$-Price of Anarchy is $n+1-\frac{v}{kc}$.
Clearly, in either case, and just as in the baseline model, the price
of anarchy is unbounded, growing linearly with $n$.

We now consider how PoA changes as a function of $k$, i.e. the number
of targets each defender has. When $k\leq \frac{v}{cn}$, a Nash
equilibrium exists and the PoA is $1$; when $\frac{v}{cn}<k\leq
\frac{v}{c}$, PoA increases linearly in $k$ with slope
$\frac{nc}{v}$. However, when $k>\frac{v}{c}$, a Nash equilibrium does
not exist and the approximate PoA is $n+1-\frac{v}{kc}$, which
increases very slowly with $k$, and is bounded by $n+1$ when $k
\rightarrow \infty$. Figure \ref{fig2} illustrates the relationship
between (approximate) Price of Anarchy and $k$ for $n=2$. 
When $k$ is very small, PoA = 1.
For intermediate $k$, PoA increases linearly, and when $k$ is
sufficiently large, Nash equilibrium no longer exists, and
$\epsilon$-PoA increases quite slowly, converging to 3 when $k
\rightarrow \infty$.

\begin{figure}
\begin{center}
\includegraphics[width=120mm,height=67.5mm]{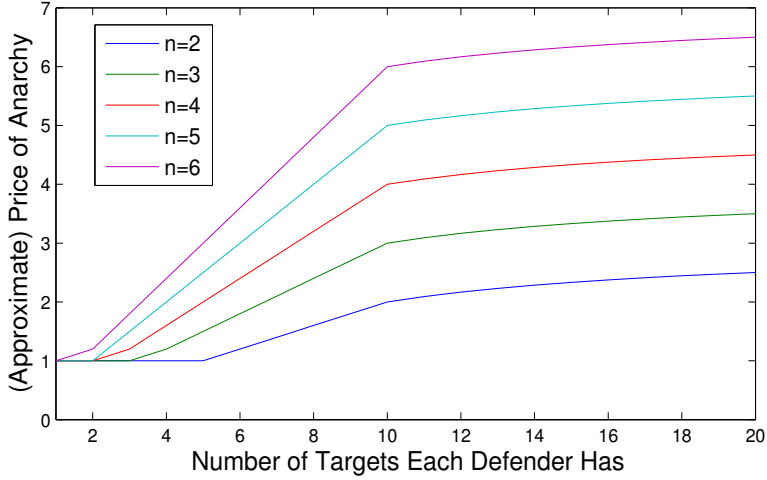}
\caption{(Approximate) Price of Anarchy when $\frac{c}{v}=0.1$}\label{fig2}
\end{center}
\end{figure}

\subsection{The General Model}

Both the baseline and the multi-target models made rather strong
assumptions about the structure of the utility functions of the
players.
In the general model, we relax these assumptions, allowing for
arbitrary utilities for the players when the target is attacked or
not, and when it is protected or not (when attacked).
Quite surprisingly, our findings here are \emph{qualitatively
  different}: the special case of the baseline and
multi-target models turns out to be an exception, rather than the rule
when more general models are considered.

Just as before, we start by characterizing Nash and approximate Nash equilibria.

\begin{restatable}{Theorem}{EquGeneral}
\label{thm:EquGeneral}
In the \emph{General model}, 
Nash equilibrium exists \emph{if and only if} $U^c-U^u\geq
kc-\frac{(n-1)(\Omega-U^c)}{n}$. 
In this equilibrium all targets are protected
with probability 1.
\end{restatable}

\begin{proof}

We firstly claim that Nash equilibrium can appear \emph{only if}
coverage probabilities of all of targets $t_{ij}$ are
identical. Otherwise, there will be a target $t_{ik}$ which has the
probability 0 of being attacked, and the defender $i$ has an incentive
to decrease $q_{ik}$. To determine a Nash equilibrium, we therefore need only consider scenarios in which all targets have the same coverage probability. 

When all targets have the same coverage probability $q$ to be
protected, the utility of each defender is $$u=\frac{(U^c-U^u-nkc)q+U^u+(nk-1)\Omega}{n}.$$
If $q<1$, then some defender $i$ could increase $q$ to $q+\delta$ for
all of her targets to ensure none of them are attacked, and obtain
utility of $u'=k\Omega-k(q+\delta)c$, so that
$$u'-u=\frac{(U^c-U^u)(1-q)+(\Omega-U^c)-nkc\delta}{n}.$$
As $U^c\geq U^u$, $\Omega\geq U^c$, and $\delta$ can be arbitrarily
small, $u'-u>0$ when $q<1$, which means that this cannot be a Nash
equilibrium. 
Thus, the only possible equilibrium can be $q_{ij} = 1$ for all
targets $t_{ij}$.

When all targets have the same coverage probability $q=1$, each
defender's utility is
$$u=\frac{U^c-nkc+(nk-1)\Omega}{n}.$$
We claim that if a defender $i$ has an incentive to deviate, it is
optimal for this defender to use the same coverage probability for all her targets. Otherwise, for some target $t_{ik}$ which has probability $0$ of being attacked, she could decrease $q_{ik}'$ to obtain higher utility.  If probabilities of targets protected by defender $i$ are all $q'$ $(0\leq q'<1)$, then her expected utility is
$u'=(U^c-U^u-c)q'+U^u+(k-1)(\Omega-q'c)$, and
$$u'-u=(U^c-U^u-kc)(q'-1)+\frac{(n-1)(U^c-\Omega)}{n}.$$
We therefore have two cases:
\begin{enumerate}[1)]
\item If $U^c-U^u\geq kc$, then $u'-u\leq 0$, and $q = 1$ for all
  targets is a Nash equilibrium.
\item If $U^c-U^u< kc$,  the maximal value of $u'-u$ corresponds to $q'=0$:
$$\max_{0\leq q'<1}u'-u= -(U^c-U^u-kc)-\frac{(n-1)(\Omega-U^c)}{n}.$$
If $kc-\frac{(n-1)(\Omega-U^c)}{n}\leq U^c-U^u<kc$, $u'-u\leq 0$, it is a Nash equilibrium; otherwise, it is not.
\end{enumerate}
To sum up, a Nash equilibrium exists \emph{if and only if} $U^c-U^u\geq
kc-\frac{(n-1)(\Omega-U^c)}{n}$, and the equilibrium corresponds to all
targets having probability 1 of being protected.
\end{proof}


Next, we characterize the optimal approximate equilibrium when no Nash
equilibrium exists.
\begin{restatable}{Theorem}{AppGeneral}
\label{thm:AppGeneral}
In the \emph{General model}, in the optimal $\epsilon$-equilibrium all
targets are protected with probability $\frac{\Omega-U^u}{kc}$. The corresponding $\epsilon$ is $\frac{(\Omega-U^u)(kc-U^c+U^u)}{cnk}$.
\end{restatable}

\begin{proof}
When all targets have the same coverage probability $q$, the expected utility of each defender is 
$$u=\frac{(U^c-U^u-nkc)q+U^u+(nk-1)\Omega}{n}.$$
Suppose $0\leq q<1$. If some defender $i$ increases $q$ to
$q+\delta_{ij}$ for target $t_{ij}$, then she would obtain utility
$u'=\sum_{j=1}^{k} \Omega-(q+\delta_{ij})c$, and
\begin{equation}\label{d1}
\begin{aligned}
u'-u&= \frac{\Omega-(U^c-U^u)q-U^u}{n}-\sum_{j=1}^{k} \delta_{ij}c\\
&\leq \frac{\Omega-(U^c-U^u)q-U^u}{n}.
\end{aligned}
\end{equation}
Now we consider scenarios in which a defender $i$ could obtain higher
utility by decreasing protection probability. 
We claim that if a defender $i$ has an incentive to deviate, it is
optimal for this defender to use the same coverage probability for all her targets. Otherwise, for some target $t_{ik}$ which has probability $0$ of being attacked, she could decrease $q_{ik}'$ to obtain higher utility. 
Thus, we need only consider cases in which a defender deviates by
decreasing coverage probabilities for all her targets to $q-\delta$. Her utility will become $u''=(U^c-U^u-kc)(q-\delta)+U^u+(k-1)\Omega$.
Since $U^c-U^u<kc$, $\delta = q$ (the maximal value of $\delta$)
maximizes $u''-u$:
\begin{equation}\label{d2}
\max_{0<\delta\leq q} u''-u= \frac{\Omega-(U^c-U^u)q-U^u}{nk}+kcq+U^u-\Omega.
\end{equation}
By comparing the value of equation (\ref{d1}) and equation (\ref{d2}), we get different values of $\epsilon$ for $\epsilon$-equilibrium:
\begin{equation}
\epsilon=
\begin{cases}
\frac{\Omega-(U^c-U^u)q-U^u}{n}, &\text{if $0\leq q\leq \frac{\Omega-U^u}{kc}$;}\\
\frac{\Omega-(U^c-U^u)q-U^u}{n}+kcq+U^u-\Omega, &\text{if $\frac{\Omega-U^u}{kc}<q\leq 1$.}
\end{cases} \nonumber
\end{equation}
When $q=\frac{\Omega-U^u}{kc}$, we get the minimal $\epsilon=\frac{(\Omega-U^u)(kc-U^c+U^u)}{cnk}$. 

We claim that the $\frac{(\Omega-U^u)(kc-U^c+U^u)}{cnk}$-equilibrium can appear \emph{only if} all targets have the same coverage probability $q$. We prove this by contradiction.  Suppose that targets have different coverage probabilities.
This gives rise to two cases: $1)$ Each defender uses an identical
coverage probability for each target she owns (these may differ
between defenders); and $2)$ Some defender has different coverage probabilities for her targets. 
In case $1)$, there exist $\beta$ defenders ($1\leq \beta<n$) who have
the same minimal coverage probability $q'$. The expected utility for
each defender among these $\beta$ is 
$$u=\frac{(U^c-U^u-k\beta c)q'+U^u+(k\beta-1)\Omega}{\beta}.$$
When $\frac{\Omega-U^u}{kc}<q'\leq 1$, some defender $i$ among these
$\beta$ could decrease the coverage probability of all her targets to
$0$ and obtain the utility of $u_1=U^u+(k-1)\Omega$, so that
\begin{equation}
\begin{aligned}
u_1-u & = \frac{\Omega-(U^c-U^u)q'-U^u}{\beta}+kcq'+U^u-\Omega \\
& > \frac{\Omega-(U^c-U^u)q'-U^u}{n}+kcq'+U^u-\Omega.
\end{aligned} \nonumber
\end{equation}
When $0\leq q'\leq \frac{\Omega-U^u}{kc}$, some defender $i$ among
these $\beta$ can increase coverage probabilities of all her targets to $q'+\delta_3$ to obtain utility of
 $u_2=k\Omega-k(q'+\delta_3)c$, with
 \begin{equation}
\begin{aligned}
u_2-u & = \frac{\Omega-(U^c-U^u)q'-U^u-k\beta c\delta_3}{\beta} \\
& > \frac{\Omega-(U^c-U^u)q'-U^u}{n},
\end{aligned} \nonumber
\end{equation} 
where the inequality holds because $\delta_3$ can be arbitrarily small. Thus, no profile in case $1)$ can be a $\frac{(\Omega-U^u)(kc-U^c+U^u)}{cnk}$-equilibrium. 
In case $2)$, any defender who has different coverage probabilities
among her targets can always increase her 
payoff by decreasing the coverage probabilities of the targets with
higher coverage to yield identical coverage for all targets. Consequently, no
profile in case $2)$ can be a
$\frac{(\Omega-U^u)(kc-U^c+U^u)}{cnk}$-equilibrium.
\end{proof}


As the final step towards characterizing the Price of Anarchy, we
derive optimal social welfare in this model.
\begin{Theorem}
In the \emph{General model}, the optimal social welfare $SW_O$ is
\begin{equation}
SW_O=
\begin{cases}
U^c-nkc+(nk-1)\Omega, &\text{if $U^c-U^u\geq nkc$;}\\
U^u+(n-1)\Omega, &\text{if $U^c-U^u<nkc$.}
\end{cases} \nonumber
\end{equation}
\end{Theorem}

\begin{proof}[Proof sketch]
We firstly claim that we could get optimal social welfare \emph{only
  if} all targets have the same coverage probability $q$. Otherwise,
some target $t_{ij}$ has probability $0$ of being attacked, and we can
decrease $q_{ij}$ to improve social welfare. Consequently, we need
only to consider an optimal symmetric coverage probability $q$ to
maximize social welfare, which can be done in a manner similar to that
for the baseline case.
\end{proof}

If $U^c-U^u\geq kc-\frac{(n-1)(\Omega-U^c)}{n}$, the Nash equilibrium is
unique, with all targets protected with probability 1. 
The corresponding social welfare is $$SW_E=U^c-nkc+(nk-1)\Omega.$$

So far we have not yet added any constrains to value of $\Omega$, $U^c$,
 and $U^u$ (except that $\Omega\geq U^c\geq U^u$). 
In order to make \emph{Price of Anarchy} well-defined, we need to add
constraints that values of  $\Omega$, $U^c$, and $U^u$ are all non-positive
(just as in the previous two models) or all non-negative. To be
consistent with previous models, we add constraints that $U^c$, $U^u$
and $\Omega$ are all non-positive (little changes if all are non-negative). 
 
In the case of a unique Nash equilibrium, the price of anarchy is
\begin{equation}
PoA=
\begin{cases}
1, &\text{if $U^c-U^u\geq nkc$;}\\
\frac{U^c-U^u-nkc}{U^u+(nk-1)\Omega}+1, &\text{if $kc-\frac{(n-1)(\Omega-U^c)}{n}\leq$} \\
& U^c-U^u< nkc.
\end{cases} \nonumber
\end{equation}

If $U^c-U^u< kc-\frac{(n-1)(\Omega-U^c)}{n}$, there is no Nash
equilibrium. 
The Social Welfare in the optimal approximate equilibrium is 
$$SW_E=(U^c-U^u-nkc)\frac{\Omega-U^u}{kc}+U^u+(nk-1)\Omega,$$ 
and the $\frac{(\Omega-U^u)(kc-U^c+U^u)}{cnk}$-Price of Anarchy is $\frac{(U^c-U^u-nkc)(\Omega-U^u)}{kcU^u+(nk-1)kc\Omega}+1$.

\begin{figure}
\begin{center}
\includegraphics[width=120mm,height=67.5mm]{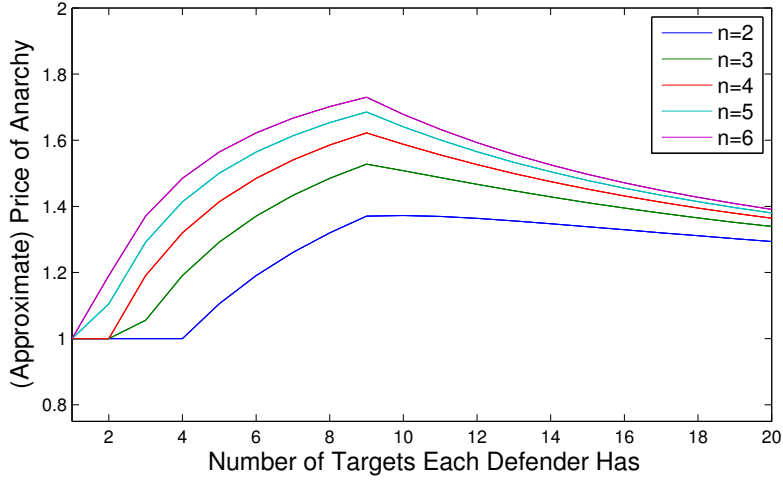}
\caption{(Approximate) Price of Anarchy when $ c=1,\Omega=-1,U^c=-2$ and $U^u=-10$}\label{fig4}
\end{center}
\end{figure}

We now analyze the relationship between ($\epsilon$-)PoA and the values
of $n$ and $k$. Here are the key differences from the Multi-Target Model. First we consider ($\epsilon$-)PoA as
the function of $n$.  
If $\Omega=0$, the result is same as that in the
Multi-Target Model: ($\epsilon$-)PoA linearly increases in $n$, and is
therefore unbounded.  
 However, if $\Omega\neq 0$, while PoA and $\epsilon$-PoA are increasing
 in $n$, as $n \rightarrow \infty$, they approach $1-\frac{c}{\Omega}$ and $1+\frac{U^u-\Omega}{k\Omega}$, respectively.
In other words, PoA (exact and approximate) is bounded by a constant,
for a constant $k$!

Consider now approximate price of anarchy as a function of $k$.
If $\Omega=0$, it is bounded by $n+1$. 
However, if
$\Omega\neq 0$, when $kc-\frac{(n-1)(\Omega-U^c)}{n}\leq U^c-U^u$, it is an
increasing function of $k$. When $kc-\frac{(n-1)(\Omega-U^c)}{n}>U^c-U^u$,
it may at first increase or decrease in $k$, depending on the the
values of the model parameters. However, when $k$ is large enough,
price of anarchy will invariably be decreasing in $k$, and as $k
\rightarrow \infty$, $\epsilon$-PoA $\rightarrow 1$.
Figure \ref{fig4} provides an example
of the relationship between $\epsilon$-PoA and $k$.
Observe that all the curves begin to decrease when $k>10$, and
they all approach 1 as $k \rightarrow \infty$.
Thus, price of anarchy in the general model is only unbounded in the
special case when $\Omega=0$, whereas when $\Omega \ne 0$, price of anarchy is
always bounded by a constant.
This observation is particularly surprising and significant considering the fact
that the baseline and simplified multi-target models are quite
natural, and seemingly innocuous, restrictions of the general case.
\parskip=0pt

\section{Analysis of Interdependent Multi-defender Security Games}
\label{sec:interdep-analysis}

We now develop and analyze a computational framework for approximating Nash equilibria in \emph{interdependent} multi-defender security games. A crucial step in computing (or approximating) a Nash equilibrium of a
game is to consider the problem of computing a best response for
an arbitrary player (in our case, defender, since the attacker's best
response is straightforward).
Next, we develop a novel mixed-integer linear programming formulation
for computing ASE best response, and then propose a hightly effective
heuristic method for approximating ASE in multi-defender games.

\subsection{Computing Defender Best Response: A Mixed-Integer Linear Programming Formulation}

While ASE seems a very natural alternative to SSE even in two-player
security games, we are not aware of any proposals for computing it.
Below, in equations~\ref{O:obj}-\ref{C:nonlin}, we present the first
(to our knowledge) mixed-integer linear programming formulation for
computing ASE which, in our case, would compute a best response for an
arbitrary defender $i$ when the strategies of all other players,
$q_{-i}$, are fixed.
\begin{align}
&\hspace*{-5pt}\max_{a,q_i,s,u,v} u - \sum_{j \in T_i}\sum_{o \in O} c_j^o q_{i,j}^o\label{O:obj}\\
\nonumber & \mathrm{s.t.}\\
&\hspace*{-5pt} 0 \le q_{i,j}^o \le 1 & \forall \ j \in T_i,\ \forall o\label{C:validStrat1}\\
&\hspace*{-5pt}\sum_{o \in O}q_{i,j}^o = 1 & \forall j \in T_i\label{C:validStrat2}\\
&\hspace*{-5pt} a_j \in \{0,1\} & \forall \ j \in T\\
&\hspace*{-5pt}\sum_{j \in T}a_j \geq 1\label{C:validStratAttack}\\
&\hspace*{-5pt} 0 \le v - \sum_o q_{i,j}^o V_{j}^o \le (1-a_j)M & \forall \ j \in T_i\label{C:chooseA1}\\
&\hspace*{-5pt} 0 \le v - \sum_o q_{-i,j}^{o} V_{j}^o \le (1-a_j)M & \forall \ j \in T_{-i}\label{C:chooseA2}\\
&\hspace*{-5pt} s_j = v - \sum_o q_{i,j}^o V_{j}^o & \forall \ j \in T_i\label{C:optAction1}\\
&\hspace*{-5pt} s_j = v - \sum_o q_{-i,j}^{o} V_{j}^o & \forall \ j \in T_{-i}\label{C:optAction2}\\
&\hspace*{-5pt} a_j + Ms_j \geq 1 & \forall \ j \in T\label{C:compOptAttack}\\
&\hspace*{-5pt} u = f(q,a),\label{C:nonlin}
\end{align}
where $M$ is a very large number and
\[
f(q,a) = \frac{\sum_{j \in T_i} a_j \sum_{o \in O} q_{i,j}^o U_{j}^o + \sum_{j \in T_{-i}} a_j \sum_{o \in O} q_{-i,j}^oU_{j}^o}{\sum_{j \in T} a_j}.
\]
While constraint~\ref{C:nonlin} is non-linear, we can linearize it 
using McCormick inequalities.
Constraints~\ref{C:validStrat1} and~\ref{C:validStrat2} ensure that
the defender's strategy is a valid probability distribution.
Constraint~\ref{C:validStratAttack}
ensures that at least one target is chosen by the attacker.
Constraints~\ref{C:chooseA1} and~\ref{C:chooseA2} compute the
optimal attacker utility $v$; alone, they ensure that this utility
corresponds to \emph{some} attack target.
Constraints~\ref{C:optAction1} and~\ref{C:optAction2} compute an
auxiliary variable $s_j$, which is $0$ if and only if attacking a
target $j$ yields an
optimal utility to the attacker.
These variables, together with constraints~\ref{C:compOptAttack} and~\ref{C:chooseA1}-\ref{C:chooseA2}
ensure that the binary variable $a_j = 1$ if and only if the attacker (weakly)
prefers to attack target $j$; that is, these jointly compute the set of
optimal attack targets.
Finally, constraint~\ref{C:nonlin} computes the expected utility to
the defender if the attacker chooses one of his most preferred targets
uniformly at random.


If $M$ is infinite and numbers can be computed to arbitrary precision,
the above formulation is correct.
In practice, of course, numerical precision and stability are an
issue, and they arise with this formulation.
Consider constraints~\ref{C:optAction1} and~\ref{C:optAction2}, which
compute $s_j$ for all targets $j$.
These require that the expected value of the target \emph{exactly}
equal the optimal attacker utility computed in
constraints~\ref{C:chooseA1}-\ref{C:chooseA2}; even a slight error
will technically violate our requirement that $s_j = 0$ at an optimal
target.
Moreover, even if the difference is, indeed, non-zero, from an
attacker's perspective it seems intuitive that sufficiently small
differences from optimal utility are ignored.
We address these problems by adding a fixed small quantity $\delta$ to the
right-hand-side of constraints~\ref{C:chooseA1},~\ref{C:chooseA2},
while subtracting it from the right-hand-side of constraint~\ref{C:compOptAttack}.
Because the constraints are interrelated, we cannot simply choose an
arbitrary $0 <\delta < 1$, but must ensure that $\delta$ and $M$ satisfy
the constraint that $M\delta < 1-\delta$.


\subsection{Approximating ASE}
Previously,~\citeA{Vorobeychik08}
presented a convergent equilibrium approximation algorithm based on
\emph{simulated annealing (SA)} that would be applicable in our setting.
They additionally showed in simulation that SA is actually outperformed by
a simple heuristic based on \emph{iterated best response (IBR)}
dynamics.
Here, we interpret IBR as a local search heuristic, with the property that if the starting point
is a Nash equilibrium, IBR will never deviate from it (i.e., Nash
equilibrium is a fixed point).
Clearly, then, the choice of a starting point can be significant for
the performance of IBR, making it natural to consider coupling it with
random restarts.
Our main contribution in this section is to present evidence
that IBR with random restarts is a highly effective equilibrium
approximation approach in our setting (and outperforms several alternatives).
This is both of broad significance, and of particular importance in
our setting, as we use this algorithm for our analyses below.

\begin{figure}[ht]
\centering
\begin{tabular}{cc}
\includegraphics[scale=0.85]{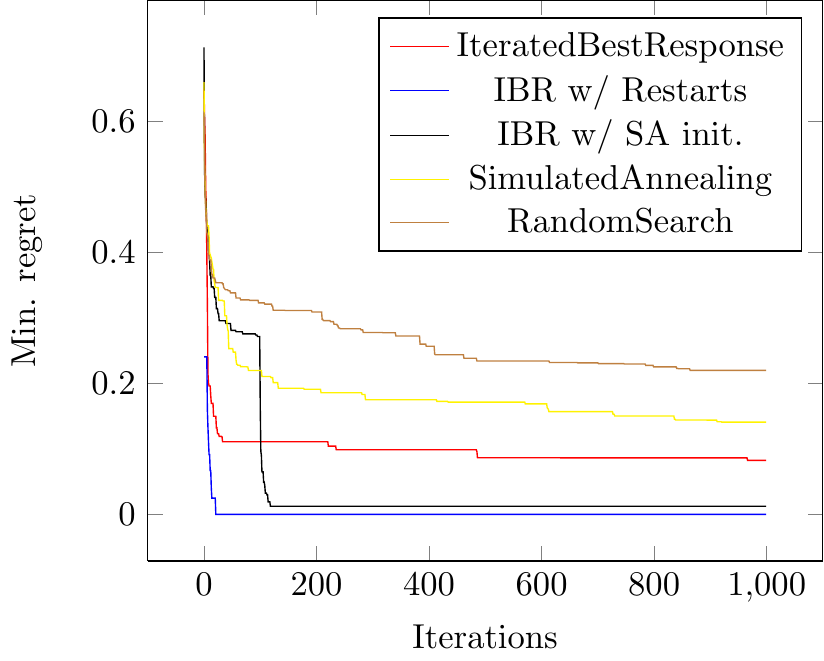}
&
\includegraphics[scale=0.85]{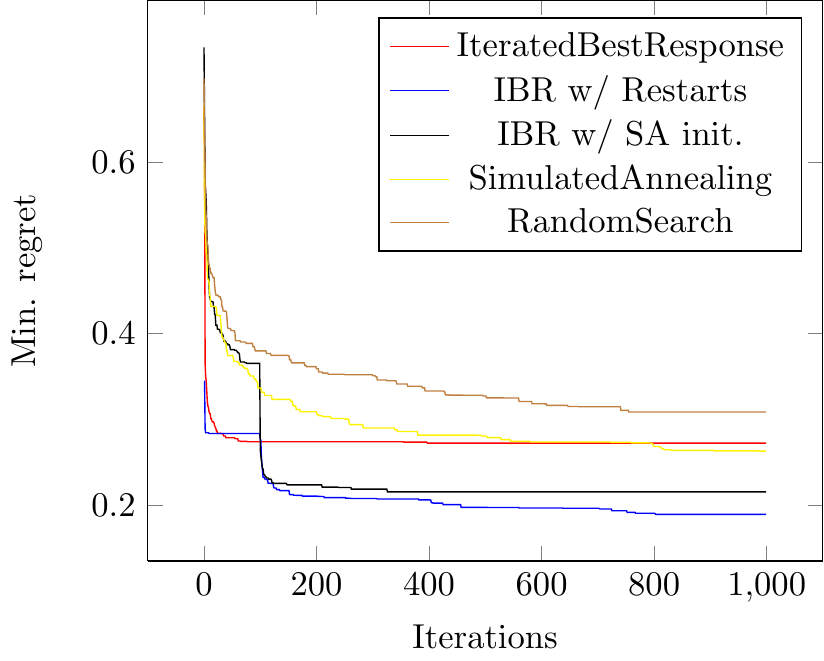}
\end{tabular}
\caption{Comparison of algorithms. Left: $\left|D\right|=2$ and
  $\left|T\right|=10$. Right: $\left|D\right|=5$ and $\left|T\right|=20$.}
\label{fig:algcompare}
\end{figure}

We compare the following Nash equilibrium approximation algorithms
executed for 1000 iterations: random search
(RS), which simply generates 1000 strategy profiles randomly, computes
the game theoretic regret of each, and chooses a profile with the
smallest regret; simulated annealing (SA), with the temperature
exponentially increasing with iterations; and iterated best response (IBR)
with no restarts.
We also include in the comparison two additional variations of IBR: the first uses SA
for the first 100 iterations, and then switches to IBR for the
remainder (starting with the best approximation produced by SA); the
second is IBR with random restarts, which we term RIBR. RIBR includes initial corner cases that may be hard to converge to in a limited amount of time (i.e., all defenders not defending, all defenders defending completely).
We execute our comparison on games with 2 players and 10 targets and
games with 5 players and 20 targets.
In all cases, targets are divided evenly among the players, and values
over the targets are generated uniformly at random.
The cost of defense is fixed at $c=0.2$, and the targets are assumed
to be independent (but players may have values for targets under
the control of other defenders).
Figure~\ref{fig:algcompare} demonstrate that in both settings, RIBR outperforms other
alternatives.

\begin{figure*}[ht!]
\centering
\begin{tabular}{ccc}
Grid & Erd\H{o}s-R\'{e}nyi & Preferential Attachment\\
\includegraphics[scale=0.19]{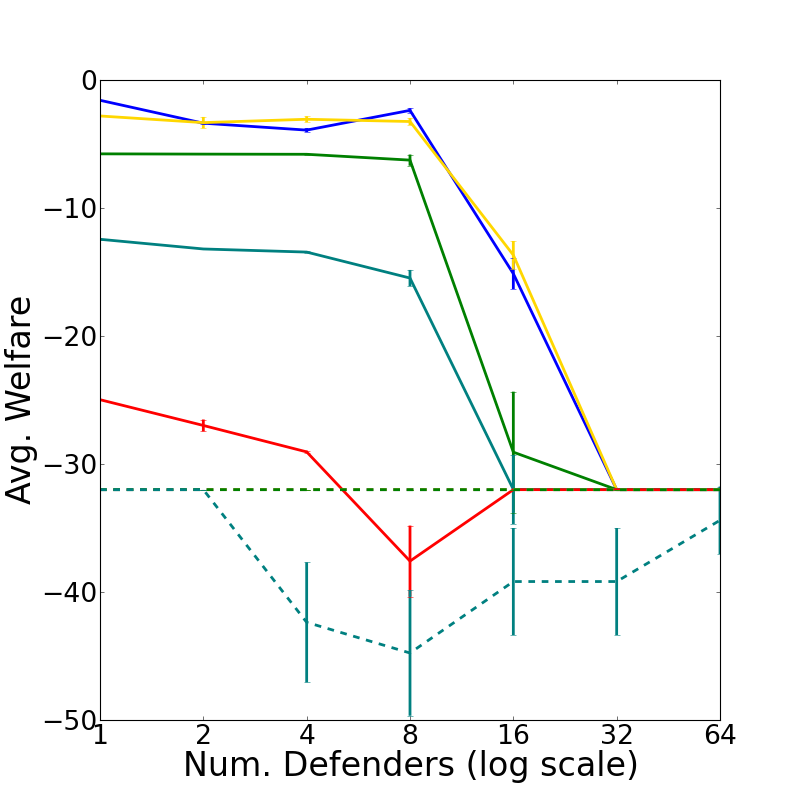}
&\includegraphics[scale=0.19]{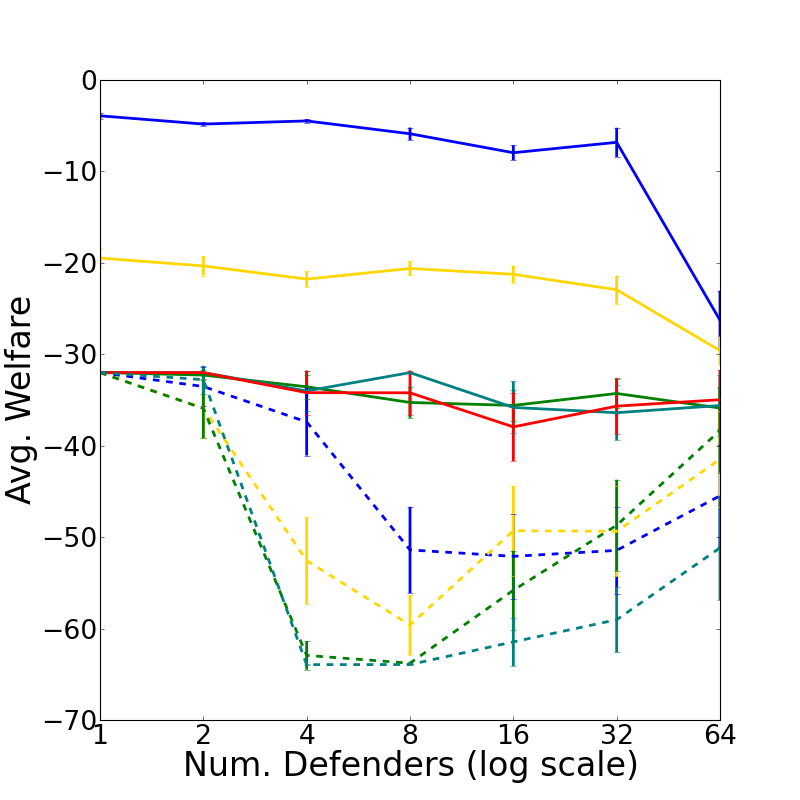}
&\includegraphics[scale=0.19]{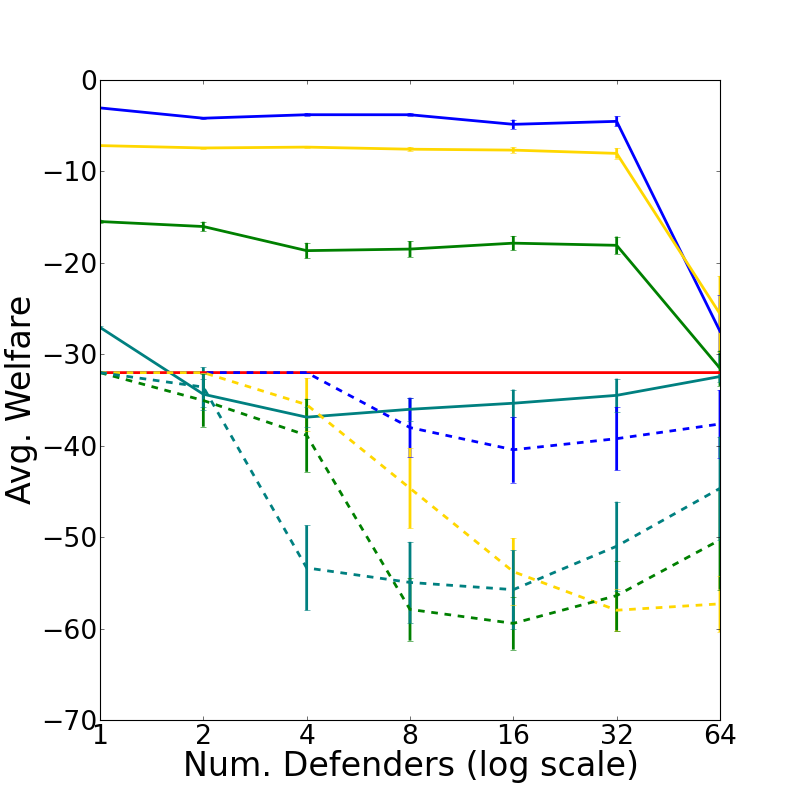}\\
\includegraphics[scale=0.19]{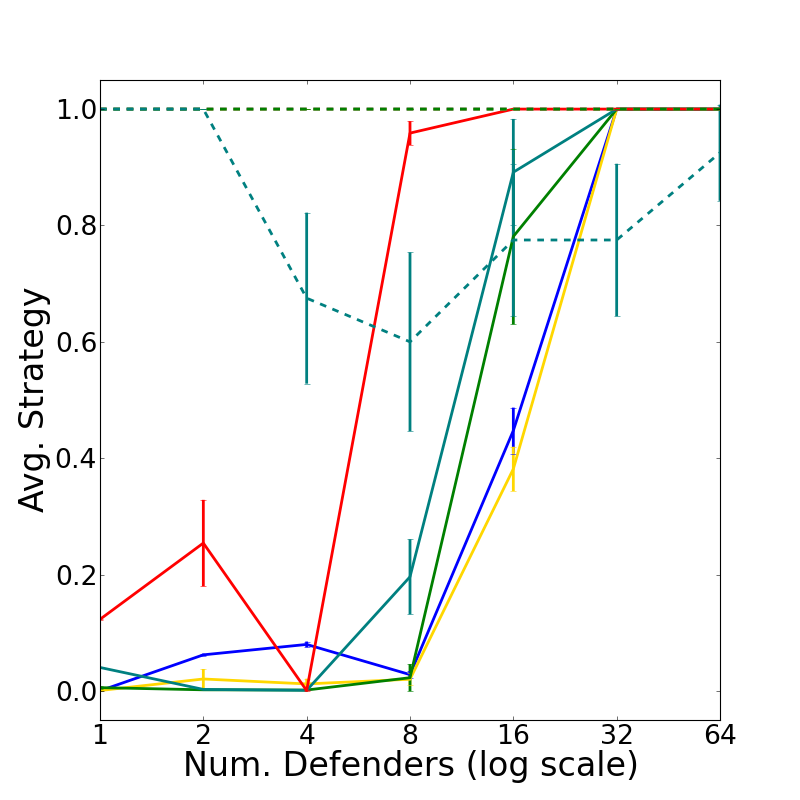} 
&\includegraphics[scale=0.19]{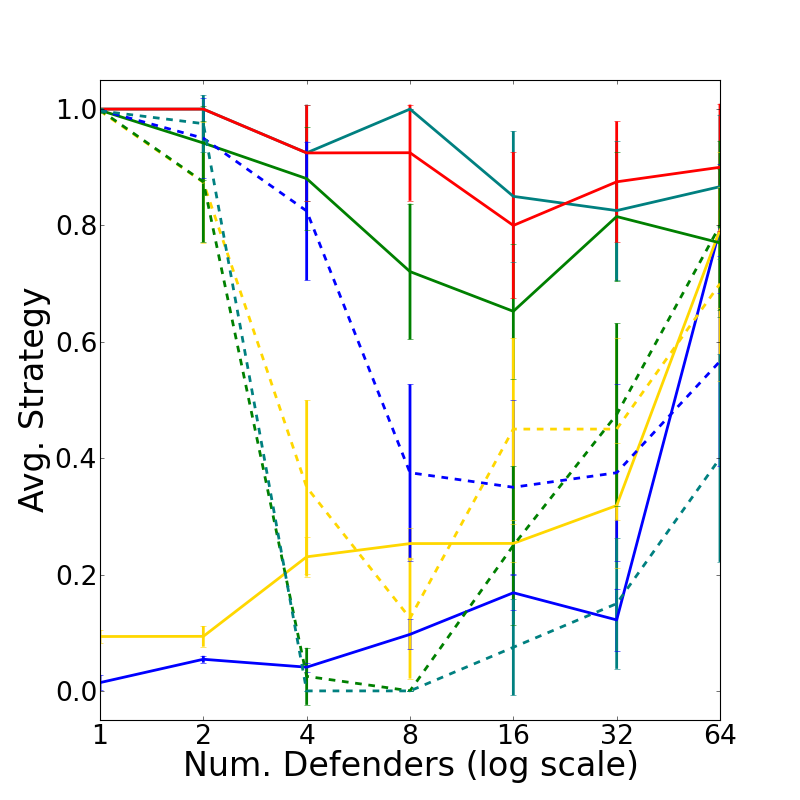} 
&\includegraphics[scale=0.19]{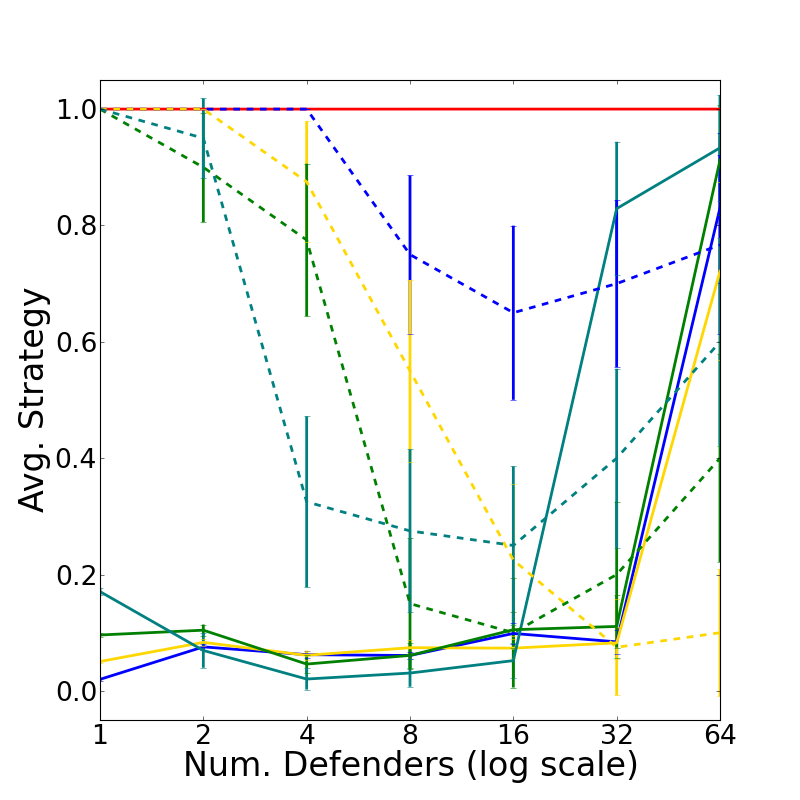}
\end{tabular}
\includegraphics[scale=0.47]{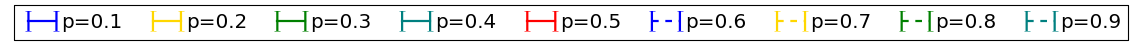}
\caption{Approximate equilibrium security outcomes for varying cascade probability $p$, as a function of (log of) the number of players (level of decentralization) for the three synthetic networks. Top: Social welfare. Bottom: Average strategy (higher strategy corresponds to higher average probability of defense). }
\label{fig:stratWelf}
\end{figure*}

\subsection{Analysis of Multi-Defender Games on Synthetic Networks}

For our first set of experiments, we use RIBR on 3 artificially
generated networks, with $40$ samples for each parameter
variation. First, we will illustrate and compare the results of our
interdependent multi-defender game on artificial networks. We use 3
commonly analyzed network structures: a grid, Erd\H{o}s-R\'{e}nyi
networks, and preferential attachment networks. In all of the
generated networks, there are $64$ nodes or targets.
For the latter two, we use the Metis graph partitioning software to
partition the nodes (targets) among defenders.
This software partitions nodes to minimize connectivity among the
targets belonging to different defenders, a property that we expect to
common hold in real networks due to efficiency considerations.

We begin by considering average strategies, as well as social welfare, for the three different synthetic networks (grid, Erd\H{o}s-R\'{e}nyi, and preferential attachment), as a function of the number of players (degree of decentralization) and the cascade probability (interdependent risk).
The results are shown in Figure~\ref{fig:stratWelf}.
We do not show these as a function of defense cost as increasing defense cost roughly mirrors decreasing cascade probability $p$.
The first rather stark observation is that network structure makes little difference when each node is controlled by a single player, but it makes a significant qualitative difference both for social welfare and actual strategies utilized by the players in all other cases.

Looking at the results in greater detail, let's consider first social welfare (Figure~\ref{fig:stratWelf}, top).
First, when interdependent risk is low ($0.1\le p \le 0.3$), social welfare follows a relatively simple pattern: increasing decentralization makes initially almost no difference, until sufficiently many players are involved, at which point social welfare falls rather dramatically; this pattern is roughly monotonic with increasing decentralization, with worst outcomes emerging when each player controls a single node, and mirrors previous findings~\cite{Vorobeychik11}.
Both Erd\H{o}s-R\'{e}nyi and preferential attachment networks are less susceptible to the negative effects of decentralization in this case than the grid network, where the dropoff occurs with fewer players (less decentralization).
This may be largely a consequence of the fact that network partitioning tools we use attempt to minimize interdependence among players---something that is likely to mirror reality---and far more opportunities for doing so exist in Erd\H{o}s-R\'{e}nyi and preferential attachment models.

When $p$ is higher (greater interdependencies), the results exhibit an entirely new phenomenology.
Across all three network models, for a sufficiently large $p$, the impact of decentralization is non-monotonic: an intermediate level of decentralization has the most detrimental impact on security, while a highly decentralized system becomes near-optimal!

\begin{figure*}[t!]
\centering
\begin{tabular}{ccc}
Power Network 1 & Power Network 2 & Power Network 3\\
\includegraphics[scale=0.19]{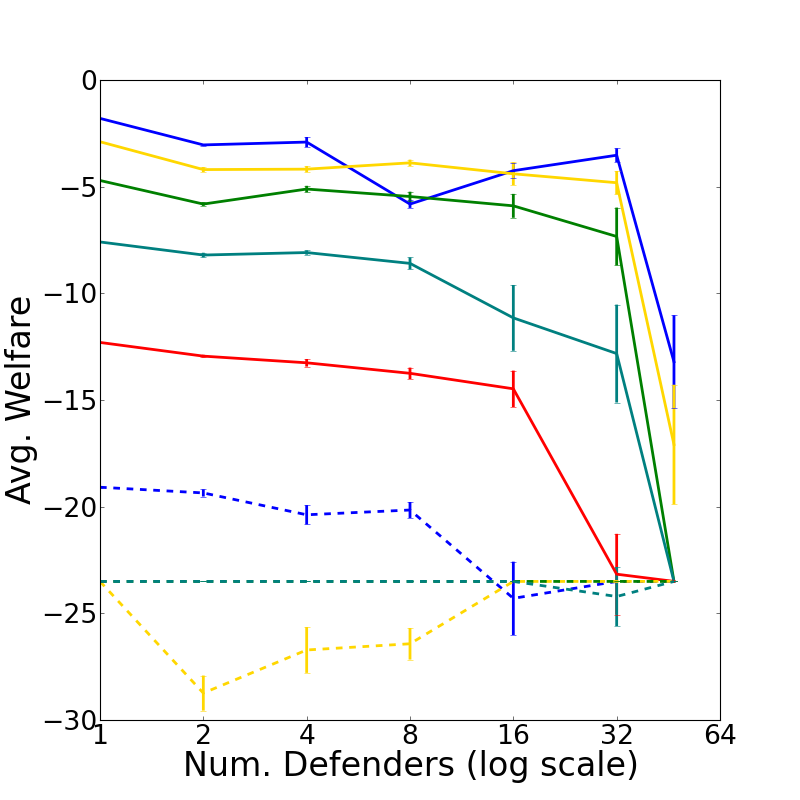}
&\includegraphics[scale=0.19]{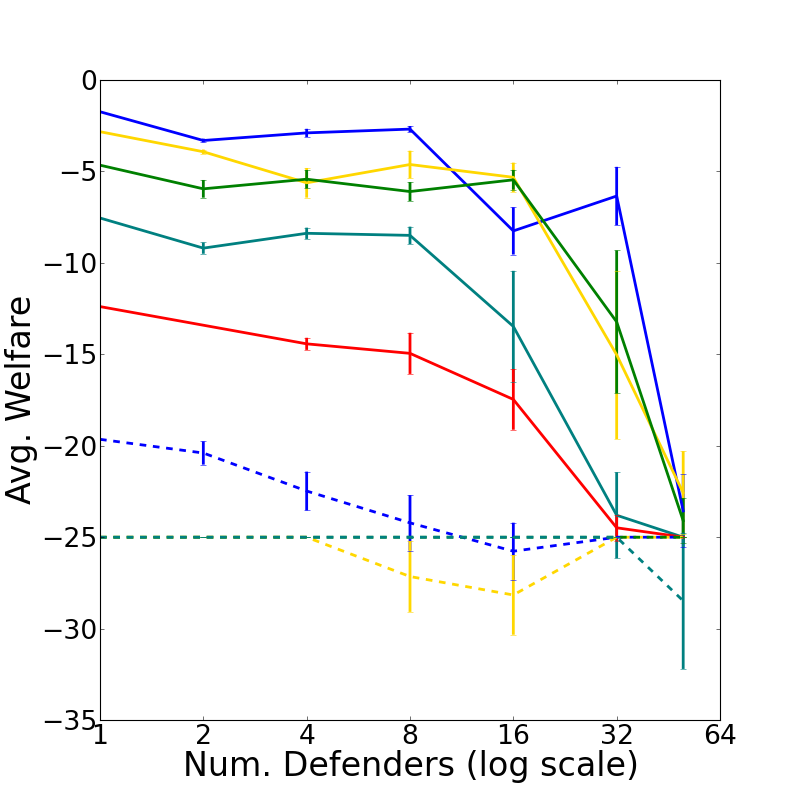}
&\includegraphics[scale=0.19]{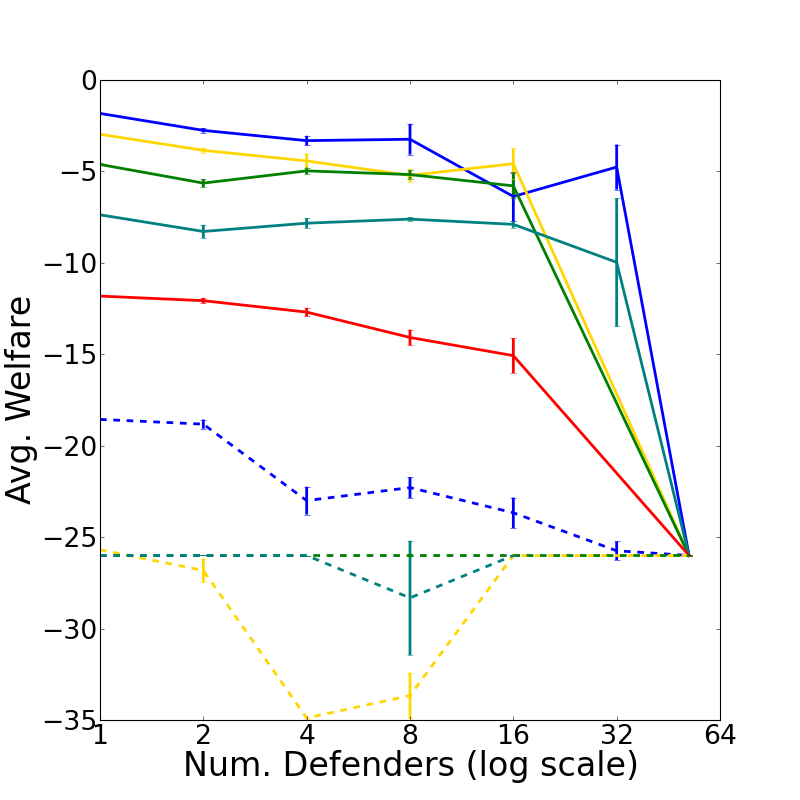}\\
\includegraphics[scale=0.19]{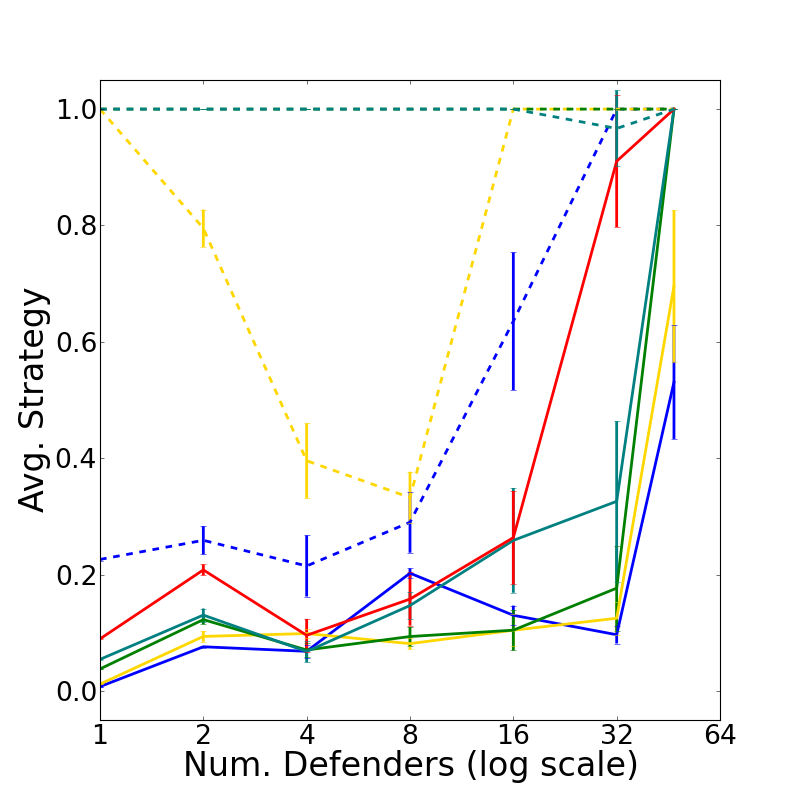} 
&\includegraphics[scale=0.19]{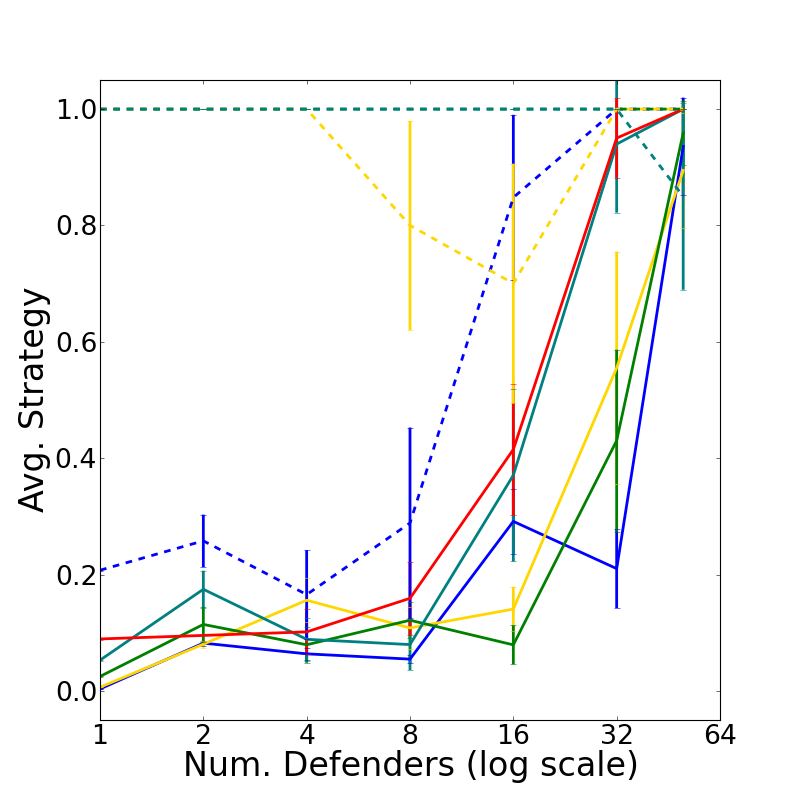} 
&\includegraphics[scale=0.19]{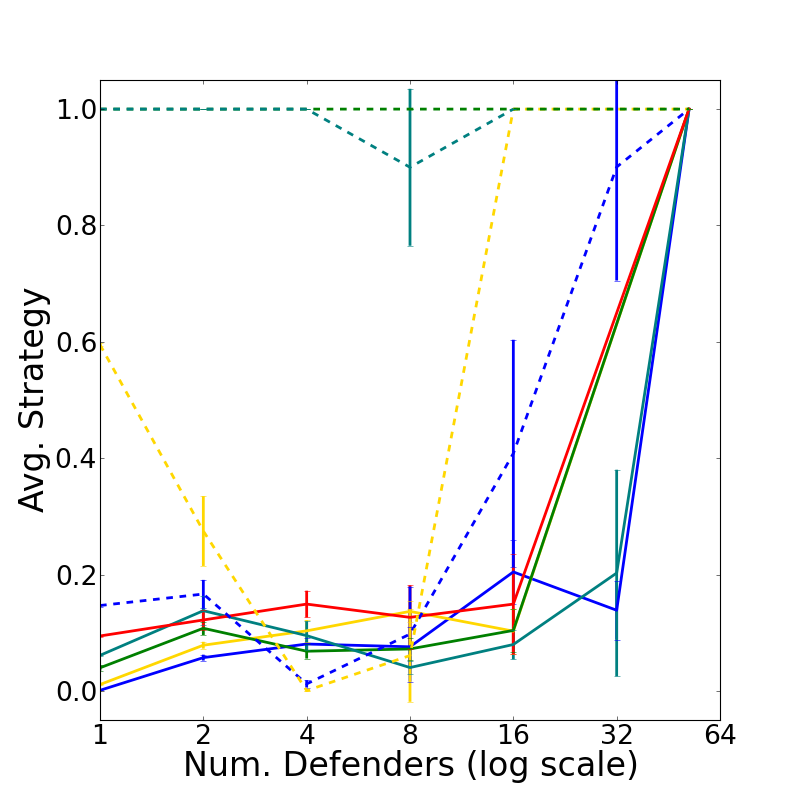}
\end{tabular}
\includegraphics[scale=0.47]{probability_legend}
\caption{Approximate equilibrium security outcomes for varying cascade probability $p$, as a function of (log of) the number of players (level of decentralization) for the three real power networks. Top: Social welfare. Bottom: Average strategy (higher strategy corresponds to higher average probability of defense). }

\label{fig:stratWelfReal}
\end{figure*}
\begin{figure*}[ht!]
\centering
\begin{tabular}{ccc}
1 player & 4 players & 64 players\\
\includegraphics[scale=0.15]{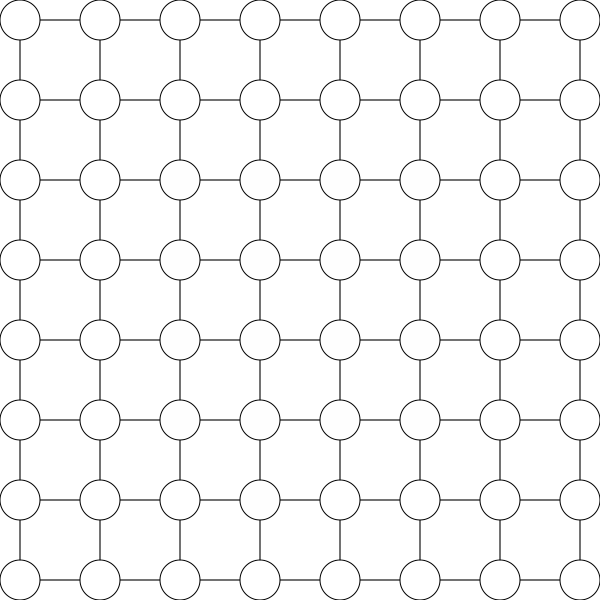}
&\includegraphics[scale=0.185]{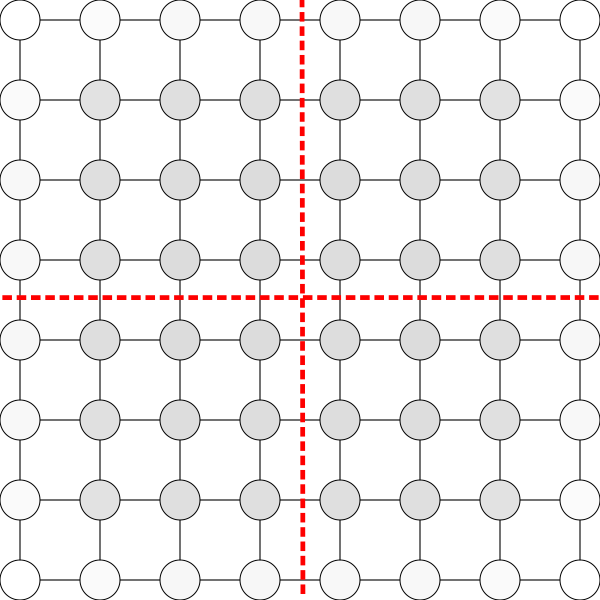}
&\includegraphics[scale=0.15]{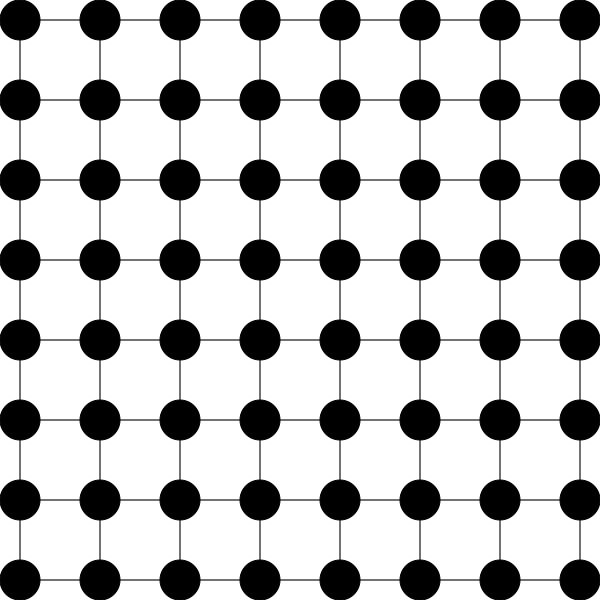}\\
\includegraphics[scale=0.15]{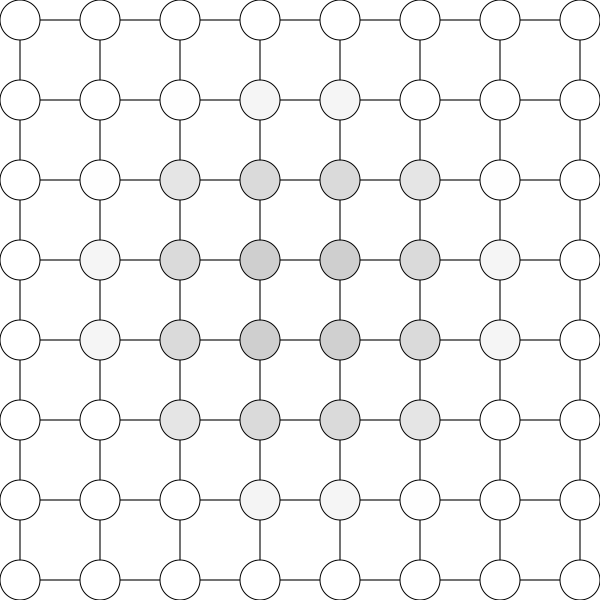}
&\includegraphics[scale=0.185]{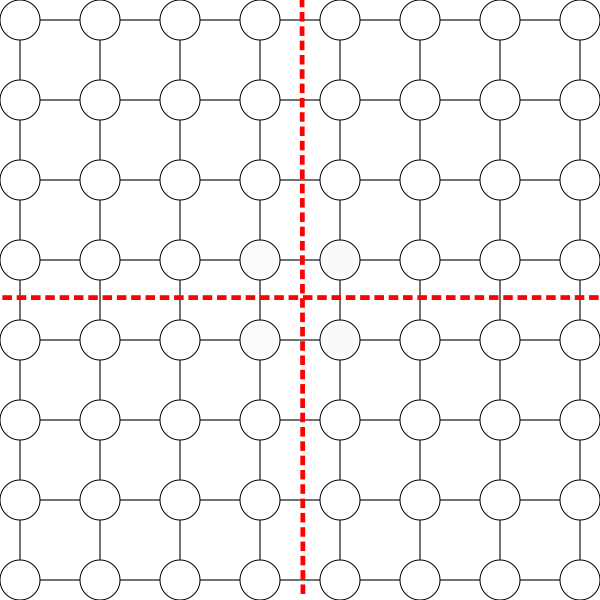}
&\includegraphics[scale=0.15]{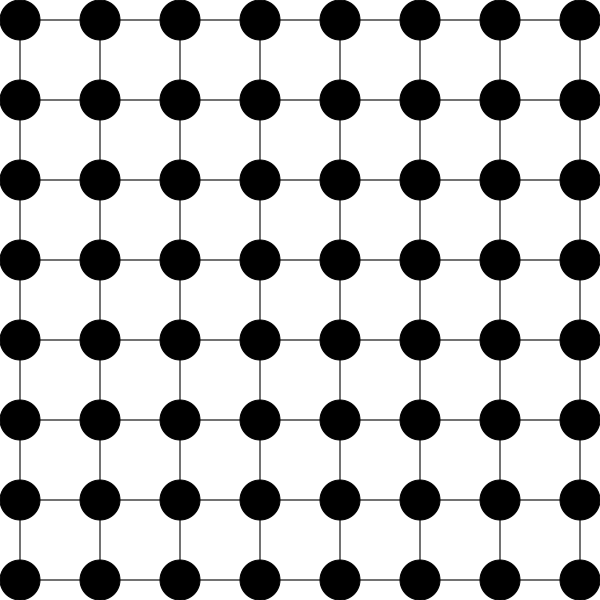}\\
\includegraphics[scale=0.15]{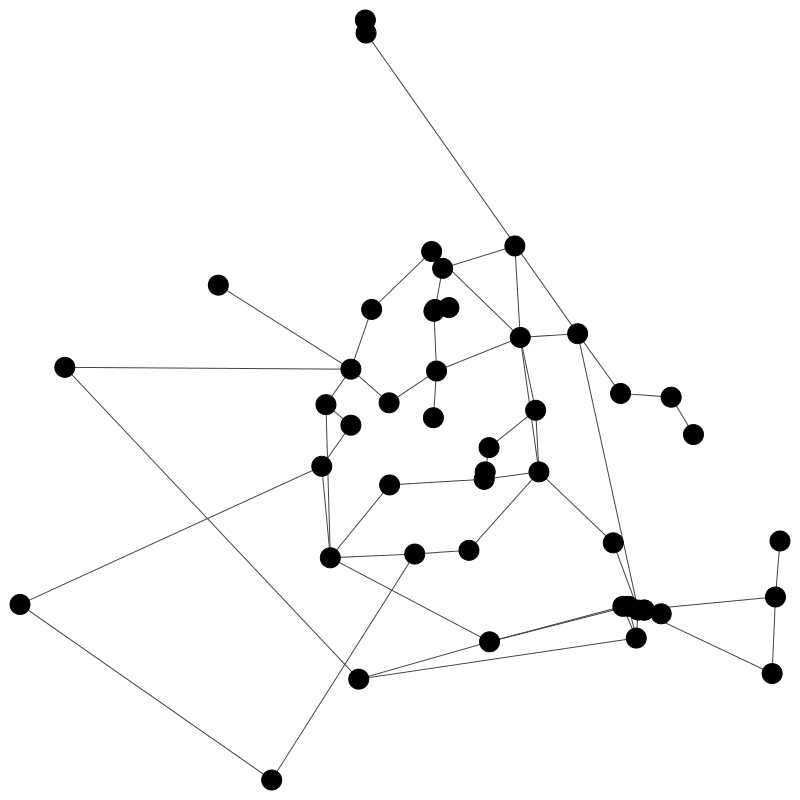}
&\includegraphics[scale=0.175]{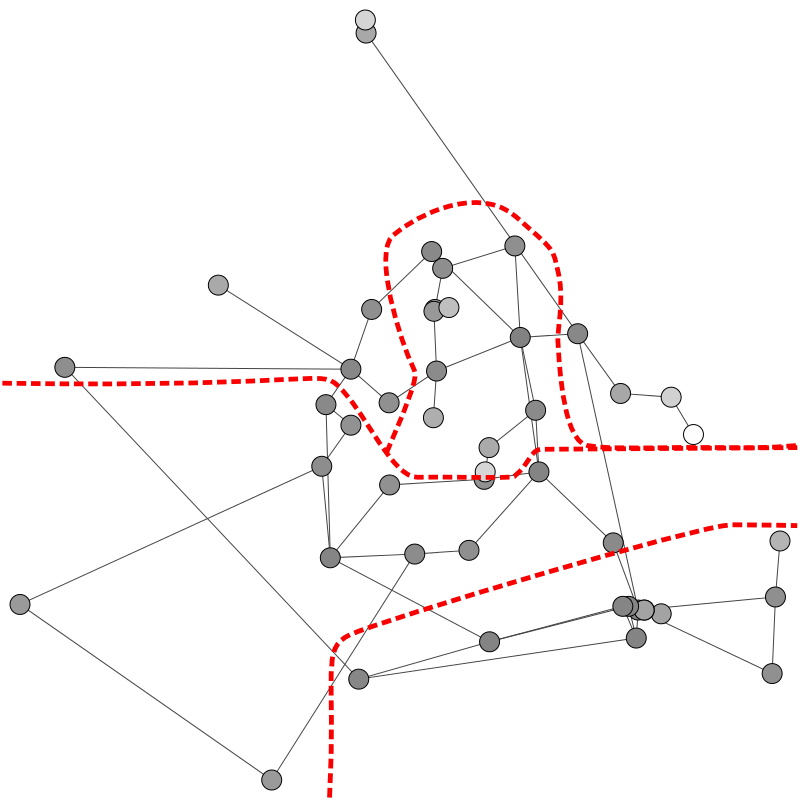}
&\includegraphics[scale=0.15]{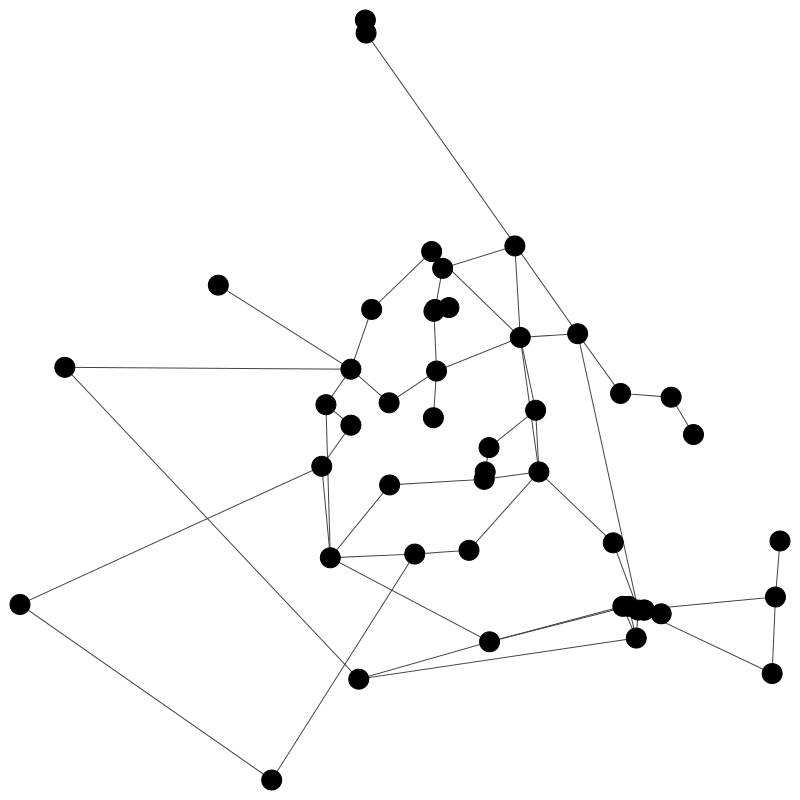}
\end{tabular}
\caption{Strategic realizations for representative games on grid and distribution network topologies. Top: Grid, $p=0.1$, middle: Grid, $p=0.4$, bottom: Power Network, $p=0.7$. Darker node colors indicate higher probability of defense (coverage); so white implies no coverage and black implies that the node is always covered (defended).  Dotted red lines indicate the partition of nodes among players (except in the rightmost case, when each player controls a single node).}

\label{fig:actualStrats}
\end{figure*}
\begin{figure*}[t!]
\centering
\begin{tabular}{ccc}
$p=0.1$ & $p=0.4$ & $p=0.7$\\
\includegraphics[scale=0.19]{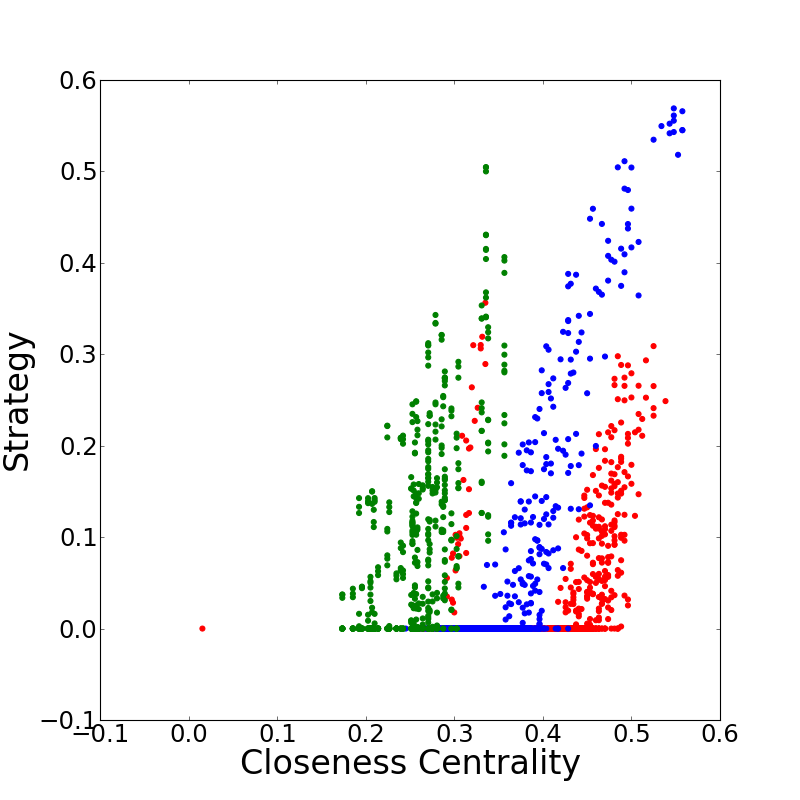}
&\includegraphics[scale=0.19]{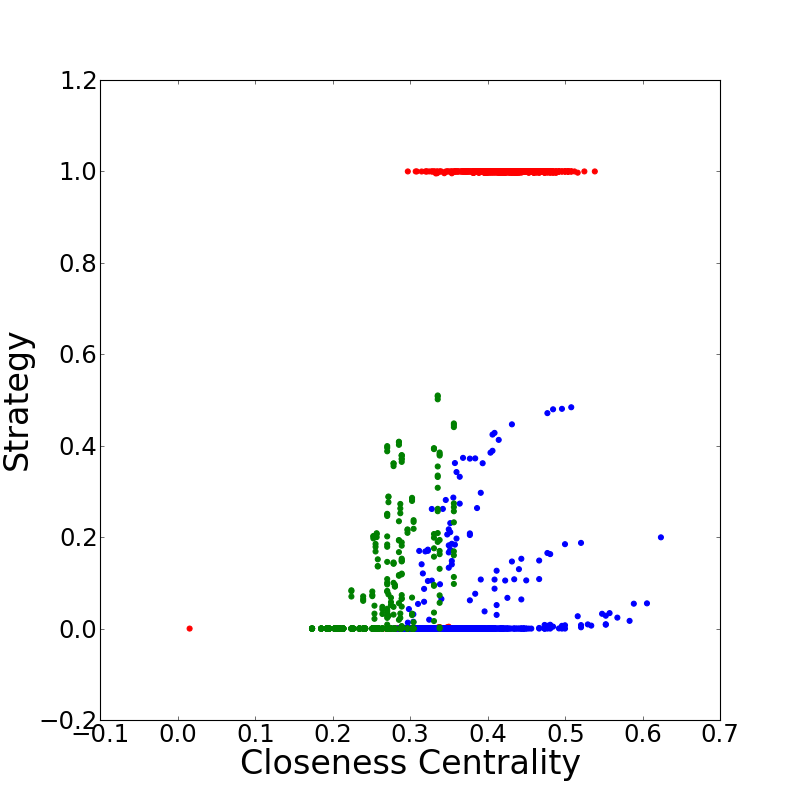} & \includegraphics[scale=0.19]{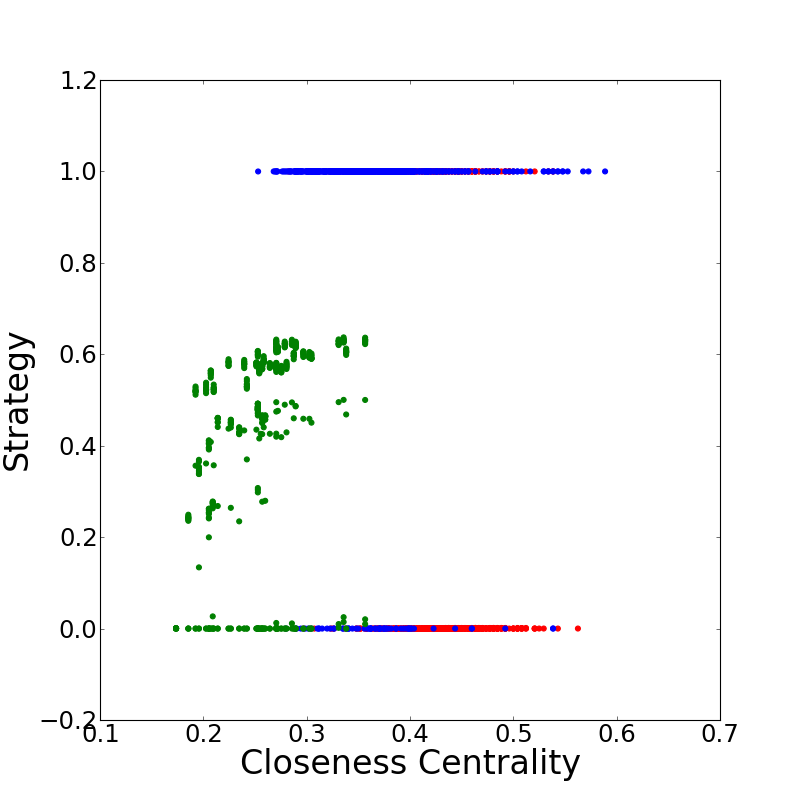} \\
\end{tabular}
\includegraphics[scale=0.47]{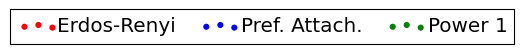}
\caption{Node security and network (closeness) centrality; each defender controls 4 nodes.}
\label{fig:centrality}
\end{figure*}
Investigating actual (average) strategic decisions by the players yields deeper insights into the findings above.
When interdependencies are weak, optimal decision is to invest relatively little in security, in any generative model.
Increasing decentralization, therefore, gives rise primarily to
over-investment, mirroring our analytical results for the limiting
case when targets were independent, 
although the tendency to over-invest is quite weak until the network is extremely decentralized, except in the grid network.
When $p$ is high, on the other hand, the predominant phenomenon is underinvestment.
This in itself is not surprising: after all, a high level of interdependencies should imply that positive externalities of security should be dominant.
What is surprising is, again, non-monotonicity in the level of decentralization: when decentralization is moderate, underinvestment can be quite dramatic.
On the other hand, a high level of decentralization often appears to dull this effect, and the level of investment in security becomes much closer to optimal.


\leaveout{
A likely reason for this phenomenon is that when decentralization is moderate, positive externalities dominate strategic considerations, since the deterrence effect of higher security on one node often pushes the attacker to attack another node belonging to the same player.
With high decentralization, however, positive and negative externalities appear to be in much greater relative balance, as higher investment pushes attacks onto other players.
}

\subsection{Results on power grid networks}
The grid network studied above is arguably the most ``artificial'', in the sense that both Erd\H{o}s-R\'{e}nyi and preferential attachment models were developed in part to resemble real networks (this is particularly true of the latter, which aims to replicate the scale-free properties of observed networks).
Surprisingly, however, the approximate equilibrium results applied to three snippets of actual power networks most resemble the phenomena observed for the grid, as can be seen in Figure~\ref{fig:stratWelfReal}.
In particular, just as in the grid above, over-investment in security appears to dominate, even at relatively high levels of interdependence, but only when decentralization is significant, while most levels of decentralization are relatively robustly near-optimal (these are, in fact, more robust to decentralization than the grid network above).

To dig somewhat deeper into the rather complex phenomenology we have observed, Figure~\ref{fig:actualStrats} shows several examples of actual strategy realizations.
First, consider the top series of plots for the grid with cascade probability $p=0.1$.
As previously described, we can clearly see that the optimal security configuration involves no security investment (leftmost grid), whereas an increasing level of decentralization gives rise to increased security investment, culminating, ultimately, with full protection in the extreme level of decentralization.
The contrast between the two extremes offers some guidance: even
though optimal global configuration involves no security, when each
player controls (and cares about) only a single node, the best
response of an attacked node is to defend it just enough to force the
attacker to attack another; for example, slightly more than the next
weakest node. 
Iterating on this idea, strategies ``cascade'' to full defense.
When the player controls more than one node, however, there is suddenly strategic tension: higher security on one node may well push the attacker to attack another node under this player's control.
Positive externalities become more significant as well: pushing the attacker to attack another node ``nearby'' is likely to gain little when cascade probabilities are high and multiple nodes owned by the defender could be affected.
For sufficiently high cascade probabilities, and sufficiently low number of players, such positive network effects can actually sway players to under-invest in security, as we can see both in the middle and last rows of Figure~\ref{fig:actualStrats} (the 4-player case).
Here, strategic complementarities make security investment not worthwhile in equilibrium: the nodes that need to be defended are relatively central, and cut across different players (i.e., the critical central nodes create a kind of ``buffer'' between defenders). This behavior diminishes as decentralization increases.

\subsection{Security and Network Centrality}

Finally, we consider the relationship between network centrality and strategic choice of a node. The results, shown in Figure~\ref{fig:centrality}, plot each node's closeness centrality and corresponding security across $10$ instances for each network. When the interdependence of networks is low, security investment appears to correlate rather strongly with closeness in synthetic networks: in other words, nodes more central in the network invest more in security. A similar relationship is seen with node degrees. This correlation, however, largely disappears with greater interdependence, likely because the difference between being one-hop vs. two-hops away from an attacked node (i.e., being a low-connected neighbor of a high-connected node) becomes considerably less in such a case.

For the power networks, the correlation between centrality and strategy remains even with higher cascade probabilities. This relationship can be seen by comparing the strategy profile in the second column of Figure~\ref{fig:actualStrats} and the closeness plot for $p=0.7$ for the power network. The structure of the power network has several small "chains" that have low centrality. The node that connects the chain to the rest of the network (higher centrality) incurs most of the defense cost, since a failure cascade would have to pass through this node to reach the rest of the chain. As the amount of interdependence increases (column 3 in Figure~\ref{fig:actualStrats}), these chains become partitioned between defenders, weakening the relationship between centrality and strategy.  

\section{Conclusion}
In this work, we have extended the current state of Stackelberg security games to include multiple defenders in non-cooperative scenarios with independent and interdependent targets. 

For the independent case, we provided complete characterizations of Nash and approximate equilibria, socially optimal solutions, and price of anarchy (PoA) for three models of varying generality. Our analysis showed that defenders generally over protect the targets, but different modelling assumptions give rise to qualitatively different outcomes: a simpler model gives rise to an unbounded PoA, whereas a more general model sees PoA converge to a constant when the number of defenders increases.

For the interdependent case, we developed a novel computation framework to overcome the difficulties of providing a concise formal analysis of such a complex model. Our simulations characterize a broad space of strategic predicaments, varying cascade probability, network structure, and system decentralization. In contrast to the independent models, our results show differing behavior in terms of security investment dependent on the strength and structure of interdependencies. One of our most stark findings is the non-monotonicity of welfare and
strategic choices as a function of the number of players: in a number
of cases, higher levels of decentralization become near-optimal, even
while intermediate decentralization leads to very poor outcomes.
As security decisions are almost universally decentralized, and often
highly interdependent, our findings enable a deeper understanding of
practical security considerations, highlighting the importance of
\emph{both}, over- and under-investment in security, and the dependence
of each on network structure, the magnitude of network externalities, and the level of decentralization. Finally, we have shown how security behavior in our model on real-world power networks relates to those in synthetic networks, highlighting similar behaviors with grid networks, and unveiling structural differences with Erd\H{o}s-R\'{e}nyi and preferential attachment networks using the relationship between strategy and centrality.


\bibliography{mp-defense}
\bibliographystyle{theapa}	

\section*{Appendix}

\EquMultiDefender*
\begin{proof}
We firstly claim that a Nash Equilibrium must have identical coverage probabilities for all targets $t_{ij}$. Otherwise, there will be a target $t_{ik}$ which has the probability 0 of being attacked, and defender $i$ has an incentive to decrease $q_{ik}$. 

When all targets have the same coverage probability $q$, the expected utility of each defender is $$u=\frac{(v-cnk)q-v}{n}.$$
If $q<1$, then some defender $i$ can increase $q$ to $q+\delta$ for all of her targets to make sure none are attacked, obtaining utility $u'=-k(q+\delta)c$, with
$$u'-u=\frac{v(1-q)-nkc\delta}{n}.$$
As $\delta$ can be arbitrarily small, $u'-u>0$ when $q<1$, and it cannot be a Nash equilibrium. 

When all targets have the same coverage probability $q=1$, the expected utility of each defender is
$u=-kc$.
If defender $i$ want to deviate, then one of her targets will be attacked.  Suppose this target is $t_{ir}$, with $q_{ir}<1$.  Then her expected utility is
$$u'=-v+q_{ir}(v-c) +\sum_{j=1,j\neq r}^{k}(-q_{ij}c),$$ and
$$u'-u=(v-c)(q_{ir}-1)+\sum_{j=1,j\neq r}^{k}(-q_{ij}c)+(k-1)c.$$
We now have two cases:
\begin{enumerate}[1)]
\item If $v<c$, then $(v-c)(q_{ir}-1)> 0$ and $\sum_{j=1,j\neq r}^{k}(-q_{ij}c)+(k-1)c\geq 0$, and the defender has an incentive to deviate, and there is no Nash equilibrium.
\item If $v\geq c$, then, because $q_{ij}\geq q_{ir}$ for all $j\neq r$, $$u'-u\leq (v-c)(q_{ir}-1)+(k-1)c-(k-1)q_{ir}c$$ and $$u'-u\leq (q_{ir}-1)(v-kc).$$
\end{enumerate}
In case 2, therefore, $u'-u \le 0$ iff $v \ge kc$, which yields the desired result.
\end{proof}

\AppMultiDefender*
\begin{proof}
When all defenders have the same coverage probability $q$, the expected utility of each defender is
$$u=\frac{(v-cnk)q-v}{n}.$$
Suppose $0\leq q<1$. If some defender $i$ increases $q$ to $q+\delta_{ij}$ for target $t_{ij}$, then she would obtain a utility of $u'=\sum_{j=1}^{k} -(q+\delta_{ij})c$, and
$$u'-u=\frac{v(1-q)}{n}-\sum_{j=1}^{k} \delta_{ij}c<\frac{v(1-q)}{n}.$$

We claim that if a defender $i$ has an incentive to deviate, it is
optimal for this defender to use the same coverage probability for all her targets. Otherwise, for some target $t_{ik}$ which has probability $0$ of being attacked, she could decrease $q_{ik}'$ to obtain higher utility. 
Therefore, we need only consider cases when a defender deviates by decreasing probabilities of all her targets to $q-\delta$. Then her utility is $u''=(v-kc)(q-\delta)-v$, and
$$u''-u=\frac{\delta n(kc-v)+v(n-1)(q-1)}{n}.$$
Since $v<kc$ when $\delta = q$ (the maximal value of $\delta$), the maximal value of $u''-u$ is
$$\max_{0<\delta\leq q} u''-u = \frac{v(1-q)}{n}+kcq-v.$$
Let $d_1=\frac{v(1-q)}{n}$, $d_2=\frac{v(1-q)}{n}+kcq-v$.
Then,
$$d_1-d_2=-kcq+v$$
When $q\leq \frac{v}{kc}$, $d_1\geq d_2$, and it is a $\frac{v(1-q)}{n}$-Nash Equilibrium; when $q>\frac{v}{kc}$, $d_1<d_2$, and it is a $(\frac{v(1-q)}{n}+kcq-v)$-Nash Equilibrium. 
To sum up, for $\epsilon$-Nash Equilibrium, 
\begin{equation}
\epsilon=
\begin{cases}
\frac{v(1-q)}{n}, &\text{if $0\leq q\leq \frac{v}{kc}$;}\\
\frac{v(1-q)}{n}+kcq-v, &\text{if $\frac{v}{kc}<q\leq 1$.}
\end{cases}\nonumber
\end{equation}
When $q=\frac{v}{kc}$, we obtain the minimal $\epsilon=\frac{v(kc-v)}{cnk}$. 

We claim that the $\frac{v(kc-v)}{cnk}$-Nash Equilibrium can appear \emph{only if} all targets have the same coverage probability $q$. Suppose targets have different coverage probabilities. This implies two possibilities: $1)$ Each defender uses identical coverage probabilities for her targets, and differences exist only between defenders; $2)$ Some defender uses diverse coverage probabilities among her own targets. 

We first consider case $1)$, in which targets may have different probabilities of bing protected, but each defender has the same probability of protecting her own targets. In this case there exist $\beta$ defenders ($1\leq \beta<n$) who have the same minimal coverage probability $q'$. The expected utility for each defender among these $\beta$ defenders is: 
$$u_e=\frac{(v-kc\beta)q'-v}{\beta}.$$
When $\frac{v}{kc}<q'\leq 1$, some defender $i$ among these $\alpha$ defenders can decrease probability of all her targets to $0$ to obtain a utility of $u_1=-v$, with
$$u_1-u_e=\frac{v(1-q')}{\beta}+(kcq'-v)>\frac{v(1-q')}{m}+(kcq'-v).$$
When $0\leq q'\leq \frac{v}{c}$, some defender $i$ among these $\beta$ defenders can = increase probability of all her targets to $q'+\delta_3$ to obtain a utility of
 $u_2=-k(q'+\delta_3)c$, with
 $$u_2-u_e=\frac{v(1-q')-kc\beta\delta_3}{\beta}>\frac{v(1-q')}{n},$$
where inequation holds because $\delta_3$ can be arbitrarily small. Consequently, case 1 cannot have a $\frac{v(kc-v)}{cnk}$-Nash Equilibrium. 

Now we consider case $2)$, in which there exists a defender who uses different coverage probabilities among her targets. As some of her targets have probability $0$ of being attacked, she could increase her payoff by decreasing probabilities of all of these targets to be as small as her target with the lowest coverage probability. It means that each coverage profile in case $2$ can always be transformed to a coverage profile in case 1 in a way that improves a defender's utility without hurting anyone else. This means that case 2 also cannot have a $\frac{v(kc-v)}{cnk}$-Nash Equilibrium. 
\end{proof}

\SocialMultiDefender*
\begin{proof}[Proof sketch]
We firstly claim that we could obtain optimal social welfare \emph{only if} all targets have the same coverage probability $q$. Otherwise, some target $t_{ij}$ has the probability $0$ of being attacked, and we can decrease $q_{ij}$ to improve social welfare. Consequently, we need only to consider an optimal identical coverage probability $q$ to obtain optimal social welfare, which can be done in a way similar to the baseline case.
\end{proof}
	
\end{document}